\documentclass[final]{siamart190516}
\usepackage{amsfonts,color,pslatex}
\usepackage{amssymb,amsmath,latexsym}
\usepackage{lipsum}

\usepackage{longtable}
\usepackage{multirow}
\usepackage{tabularx}
\usepackage{lipsum}
\usepackage{multicol}
\usepackage{amssymb}
\usepackage{amsmath}
\usepackage{longtable}
\newtheorem{example}{Example}

\usepackage{xcolor}

\newcommand{\GF}[2][\cc]{{\mathbb F}_{{#1}^{#2}}}

\newcommand{\GFm}{\@ifstar{\GF m^\star}{\GF m}}

\newcommand{\wt}{\mathrm{wt}}

\newcommand{\Rmnum}[1]{\expandafter\@slowromancap\romannumeral #1@}

\makeatother

\makeatletter
\newcommand{\figcaption}{\def\@captype{figure}\caption}
\newcommand{\tabcaption}{\def\@captype{table}\caption}
\makeatother

\title{Codes, differentially $\delta$-uniform functions and $t$-designs\thanks
{The research of C. Tang was supported by National Natural Science Foundation
of China (Grant No. 11871058) and China West Normal University (14E013, CXTD2014-4
and the Meritocracy Research Funds).
The research of C. Ding was supported by The Hong Kong Research Grants Council, Project No. 16300418.
The research of M. Xiong was supported by The Hong Kong Research
Grants Council, Project No. NHKUST619/17.}}

\author{Chunming Tang \thanks{School of Mathematics  and Information,
China West Normal University, Nanchong 637002,  China, and also the Department of Mathematics, The Hong Kong University
of Science and Technology, Clear Water Bay, Kowloon, Hong Kong (\tt tangchunmingmath@163.com).}
\and Cunsheng Ding\thanks{Department of Computer Science and Engineering,
        The Hong Kong University of Science and Technology, Clear Water Bay,
        Kowloon, Hong Kong ({\tt cding@ust.hk}).}
\and    Maosheng Xiong\thanks{Department of Mathematics,
        The Hong Kong University of Science and Technology, Clear Water Bay,
        Kowloon, Hong Kong ({\tt mamsxiong@ust.hk}).}
        }

\begin{document}

\maketitle

\begin{abstract}
Special functions, coding theory and $t$-designs have close connections and interesting interplay.
A standard approach to constructing $t$-designs is the use of linear codes with certain regularity.
The Assmus-Mattson Theorem and the automorphism groups are two ways for proving that a code has
sufficient regularity for supporting $t$-designs. However, some linear codes hold $t$-designs,
although they do not satisfy the conditions in the Assmus-Mattson Theorem and do
not admit a $t$-transitive or $t$-homogeneous group as a subgroup of their automorphisms. The
major objective of this paper is to develop a theory for explaining such codes and obtaining
such new codes and hence new $t$-designs. To this end, a general theory for
punctured and shortened codes of linear codes supporting $t$-designs is established,
a generalized Assmus-Mattson
theorem is developed, and a link between $2$-designs and differentially $\delta$-uniform functions
and $2$-designs is built. With these general results, binary codes with new parameters and known
weight distributions are obtained, new $2$-designs and Steiner system $S(2, 4, 2^n)$ are produced
in this paper.
\end{abstract}

\begin{keywords}
Assmus-Mattson theorem, bent function,  differentially $\delta$-uniform function, linear code,  $t$-design.
\end{keywords}

\begin{AMS}
94B05, 05B05, 06E30.
\end{AMS}


\section{Introduction}\label{Sec-introduct}

We start with a brief recall of $t$-designs. Let $\mathcal P$ be a set of $\nu$ elements
and $\mathcal B$ a multiset of $b$ $k$-subsets of $\mathcal P$, where $\nu \ge 1$,
$b \ge 0$ and  $1 \le k \le \nu$.  Let $t$ be a positive integer satisfying $1 \le t \le \nu$.
The pair $\mathbb  D=(\mathcal P, \mathcal B)$ is called a  $t$-$(\nu, k, \lambda )$ \emph{design}, or simply \emph{$t$-design},
if  every $t$-subset of $\mathcal P$ is contained in exactly $\lambda$ elements of $\mathcal B$.  The elements of $\mathcal P$ are called \emph{points},
and those of $\mathcal B$ are referred to as \emph{blocks}.

When $\mathcal B=\emptyset$, i.e., $b=0$, we put $\lambda=0$ and call
$(\mathcal P, \emptyset)$ a $t$-$(\nu, k, 0)$ design for any $t$ and $k$
with $1 \leq t \leq \nu$ and $0 \leq k \leq \nu$. A $t$-$(\nu, k, \lambda)$
design with $t > k$ must have $\lambda=0$ and must be the design
$(\mathcal P, \emptyset)$. These designs are called trivial designs. We have these conventions for the easiness of
description in the sequel. A $t$-$(\nu, k, \lambda)$ design $(\mathcal P, \mathcal B)$
is also said to be trivial if every $k$-subset of $\mathcal P$ is a block.

A $t$-design is called \emph{simple} if $\mathcal B$ does not contain repeated blocks.
 A $t$-$(\nu, k, \lambda)$ design is called a \emph{Steiner system} and denoted by
$S(t,k,\nu)$ if $t \ge 2$ and $\lambda=1$. The parameters of a $t$-$(\nu,k, \lambda )$ design satisfy:
\begin{align*}
\binom{\nu}{t} \lambda =\binom{k}{t} b.
\end{align*}

Let $\mathrm{GF}(q)$ denote the finite field with $q$ elements, where
$q$ is a prime power.
A linear code $\mathcal C$ over $\mathrm{GF}(q)$ may induce a $t$-design under certain conditions, which is formed by
the supports of  codewords of a fixed Hamming weight in $\mathcal C$.
 Let
$\mathcal P(\mathcal C)=\{0,1, \dots, \nu-1\}$ be the set of the coordinate positions of $\mathcal C$, where $\nu$ is the length of $\mathcal C$.
For a codeword $\mathbf c =(c_0, \dots, c_{\nu-1})$ in $\mathcal C$, the \emph{support} of  $\mathbf c$
is defined by
\begin{align*}
\mathrm{Supp}(\mathbf c) = \{i: c_i \neq 0, i \in \mathcal P(\mathcal C)\}.
\end{align*}
Let $\mathcal B_{w}(\mathcal C)
=\frac{1}{q-1}\{\{   \mathrm{Supp}(\mathbf c): wt(\mathbf{c})=w
~\text{and}~\mathbf{c}\in \mathcal{C}\}\}$, here and hereafter $\{\{\}\}$ is the multiset notation and
$\frac{1}{q-1}S$
denotes the multiset obtained after dividing the multiplicity of each element in the multiset $S$
by $q-1$.  For some special $\mathcal C$, $\left (\mathcal P(\mathcal C),  \mathcal B_{w}(\mathcal C) \right)$
is a $t$-design. If  $\left (\mathcal P(\mathcal C),  \mathcal B_{w}(\mathcal C) \right)$ is a $t$-design for any $0\le w \le \nu$,
we say that the code $\mathcal C$ \emph{supports $t$-designs}. Notice that such design
$\left (\mathcal P(\mathcal C),  \mathcal B_{w}(\mathcal C) \right)$ may have repeated
blocks or may be simple or trivial.

With this approach, many $t$-designs have been obtained from linear codes
\cite{AK92,Ding18dcc,Ding18jcd,DLX17,HKM04,HMT05,KM00,MT04,Ton98,Ton07}.
A major approach to constructing $t$-designs from codes is the use of the
Assmus-Mattson Theorem \cite{AM74,HP10}.
Another major approach to constructing $t$-designs from linear codes is the use of linear codes with $t$-homogeneous or $t$-transitive automorphism groups \cite[Theorem 4.18]{Dingbk18}. Interplay between codes and designs could
be found in \cite{AK92,AK98,AM74,Ding15,Ding18dcc,Ding18jcd,Dingbk18,DL17,DLX17,HP10,KM00,MS77,MT04,Ton98,Ton07}.

In 2018, Ding, Munemasa and Tonchev \cite{DMT18} introduced a family of binary linear codes   based on bent vectorial functions. These codes
 support   $2$-designs, although they do not satisfy the conditions of the Assmus-Mattson theorem, and do not admit $2$-transitive
or $2$-homogeneous  automorphism groups in general. Recently, Tang, Ding and Xiong \cite{TDX19} proved that some ternary codes,
which do not satisfy the conditions of the Assmus-Mattson theorem and do not admit 2-transitive or 2-homogeneous automorphism groups in general, hold 2-designs. These works motivate us to develop a theory that can in one strike explain why these codes support $t$-designs on one hand, and may give new $t$-designs on the other hand.

In this paper, we first determine the parameters of some shortened and punctured codes of some codes supporting $t$-designs
and pay special attention to the codes from bent functions and bent vectorial functions.
Next, we give a characterization of  codes supporting $t$-designs via the weight distributions of
their shortened and punctured codes. Further, we present a generalization of the Assmus-Mattson theorem, which  provides
a unified explanation of the codes supporting $2$-designs in \cite{DMT18} and \cite{TDX19}.
Finally, we present  a design-theoretical characterization of differentially two-valued functions.
Based on the established results, we use special differentially two-valued functions to give new binary linear codes, which hold $2$-designs but do not satisfy
the conditions of the Assmus-Mattson theorem and do not admit $2$-transitive or $2$-homogeneous automorphism groups in general.

The rest of this paper is arranged as follows. Section \ref{sec:intr} introduces  definitions and results related to
linear codes, $t$-designs and differentially $\delta$-uniform functions. Section \ref{sec:shpu} investigates shortened and punctured
codes of some linear codes supporting $t$-designs. Section \ref{sec:chde} gives a characterization of  codes supporting $t$-designs
by means of their shorted and punctured codes. Section \ref{sec:geam} presents a generalization of the Assmus-Mattson theorem.
Section \ref{sec:dede} gives a design-theoretical characterization of differentially two-valued functions and presents
new codes that do not satisfy
the conditions of the Assmus-Mattson theorem and do not admit $2$-transitive or $2$-homogeneous automorphism groups in general,
but nevertheless hold $2$-designs.
Section \ref{sec:conc} concludes this paper and makes
concluding remarks.

\section{Preliminaries}\label{sec:intr}

In this section, we briefly recall some results  on the Pless power moments of linear codes, $t$-designs, differentially $\delta$-uniform functions, and shortened and punctured codes.

\subsection{The Pless power moments and the Assmus-Mattson theorem}

Let $\mathcal C$ be a $[\nu,m,d]$ \emph{linear code} over the finite field $\mathrm{GF}(q)$, where
$q$ is a prime power.  Denote by
$(A_0, A_1, \dots, A_\nu)$ and $(A_0^{\perp}, A_1^{\perp}, \dots, A_\nu^{\perp})$
 the weight distributions of $\mathcal C$ and its dual $\mathcal C^{\perp}$,
 respectively.
 The \emph{Pless power moments} \cite{HP10} are given by
 \begin{align}\label{eq:PPM}
 \sum_{i=0}^\nu  i^t A_i= \sum_{i=0}^t (-1)^i A_i^{\perp}   \left [  \sum_{j=i}^t   j ! S(t,j) q^{m-j} (q-1)^{j-i} \binom{\nu-i}{\nu -j}  \right ],
 \end{align}
where $0\le t \le \nu$ and $S(t,j)=\frac{1}{j!} \sum_{i=0}^j  (-1)^{j-i} \binom{j}{i} i^t$.
These power moments can be employed to prove the following theorem \cite[Theorem 7.3.1]{HP10}.
\begin{theorem}\label{eq:wt-wtd}
Let $S\subseteq \{1,2,\dots, \nu\}$ with $\# S =s$. Then the weight distributions of $\mathcal C$ and $\mathcal C^{\perp}$
are uniquely determined by $A_1^{\perp}, \dots, A_{s-1}^{\perp}$ and the $A_i$ with $i\not \in S$.
These values can be found from the first $s$ equations in (\ref{eq:PPM}).
\end{theorem}

The following is a general version of the Assmus-Mattson Theorem.

\begin{theorem}\label{thm-AMTheoremExt}
Let $\mathcal C$ be a linear code over $\mathrm{GF}(q)$ with length $\nu$ and minimum weight $d$.
Let $\mathcal C^{\perp}$  with minimum weight $d^{\perp}$ denote the dual code of $\mathcal C$. Let  $t~(1\le t <\min \{d, d^{\perp}\})$  be an integer such that there are at most $d^{\perp}-t$ weights of $\mathcal C$ in $\{1,2, \ldots, \nu-t\}$.
Then $(\mathcal P(\mathcal C), \mathcal B_{k}(\mathcal C))$ and $(\mathcal P(\mathcal C^{\perp}), \mathcal B_{k}(\mathcal C^{\perp}))$
are  $t$-designs for all $k\in \{0, 1, \ldots, \nu\}$.
 \end{theorem}

Notice that some of the designs in Theorem \ref{thm-AMTheoremExt} may have \textbf{repeated blocks} or may
be \textbf{trivial} in the senses defined in Section \ref{Sec-introduct}. The following lemma provides a criterion for
obtaining a simple block set $\mathcal B_k(\mathcal C)$  \cite[Lemma 4.1]{Dingbk18}.

\begin{lemma}\label{lem-simpled}
Let $\mathcal C$ be a linear code over $\mathrm{GF}(q)$ with length $\nu$ and minimum weight $d$. Let $w$ be the largest integer with $w \le \nu$ satisfying
\begin{align*}
w-\left \lfloor \frac{w+q-2}{q-1} \right  \rfloor <d.
\end{align*}
Then there are no repeated blocks in $\mathcal B_k(\mathcal C)$ for any $d\le k \le w$. Such a block set is
said to be simple.
\end{lemma}

Combining Theorem \ref{thm-AMTheoremExt} and Lemma \ref{lem-simpled}, one obtains the
following Assmus-Mattson Theorem for constructing \textbf{simple} $t$-designs \cite{AM69}.

\begin{theorem}\label{thm-AMTheorem}
Let $\mathcal C$ be a linear code over $\mathrm{GF}(q)$ with length $\nu$ and minimum weight $d$.
Let $\mathcal C^{\perp}$  with minimum weight $d^{\perp}$ denote the dual code of $\mathcal C$. Let  $t~(1\le t <\min \{d, d^{\perp}\})$  be an integer such that there are at most $d^{\perp}-t$ weights of $\mathcal C$ in the range $\{1,2, \ldots, \nu-t\}$.
Then the following holds:
\begin{itemize}
\item $(\mathcal P(\mathcal C), \mathcal B_{k}(\mathcal C))$ is a simple $t$-design provided that
      $A_k \neq 0$ and $d \leq k \leq w$, where $w$ is defined to be the largest integer
      satisfying $w \leq \nu$ and
      $$
      w-\left\lfloor \frac{w+q-2}{q-1} \right\rfloor <d.
      $$
\item $(\mathcal P(\mathcal C^{\perp}), \mathcal B_{k}(\mathcal C^{\perp}))$ is a simple $t$-design provided that
      $A_k^\perp  \neq 0$ and $d^\perp \leq k \leq w^\perp$, where $w^\perp$ is defined to be the largest integer
      satisfying $w^\perp \leq \nu$ and
      $$
      w^\perp-\left\lfloor \frac{w^\perp+q-2}{q-1} \right\rfloor <d^\perp.
      $$
\end{itemize}
\end{theorem}

\subsection{Shortened codes and punctured codes}

Let $\mathcal  C$ be a $[\nu,m,d]$ linear code over $\mathrm{GF}(q)$ and $T$ a set of $t$
coordinate positions in $\mathcal  C$.
We use $\mathcal  C^T$ to denote the code obtained by puncturing $\mathcal  C$  on
$T$, which is called the  \emph{punctured code}  of $\mathcal C$ on $T$.
Let $\mathcal C(T)$ be the subcode of $\mathcal C$, which is the set of codewords which are
$\mathbf{0}$ on $T$.
We now puncture $\mathcal C(T)$ on $T$, and obtain a linear code $\mathcal C_{T}$, which is called the \emph{shortened code} of $\mathcal C$ on $T$.
We will need the following result on the punctured and shortened codes of $\mathcal{C}$  \cite[Theorem 1.5.7]{HP10}.

\begin{lemma}\label{lem:C-S-P}
Let $\mathcal C$ be a $[\nu,m,d]$ linear code over $\mathrm{GF}(q)$ and  $d^{\perp}$  the minimum distance of $\mathcal  C^{\perp}$.
Let $T$ be any set of $t$ coordinate positions. Then

(1) $\left ( \mathcal C_{T} \right )^{\perp} = \left ( \mathcal C^{\perp}  \right)^T$ and $\left ( \mathcal C^{T} \right )^{\perp} = \left ( \mathcal C^{\perp}  \right)_T$.

(2) If $t<\min \{d, d^{\perp} \}$, then the codes $\mathcal C_{T}$ and $\mathcal C^T$ have  dimension  $m-t$  and $m$, respectively.
\end{lemma}

\subsection{Combinatorial t-designs and their intersection numbers}

Let $\mathbb D=(\mathcal P, \mathcal B)$ be a $t$-$(\nu, k, \lambda)$ design. Let $T_0$ and $T_1$ be two disjoint subsets of $\mathcal P$
with $\# T_0=t_0$ and $\# T_1=t_1$. Denote by $\lambda_{T_1}^{T_0}$ the number of blocks in $\mathcal B$ that contain $T_1$ and are disjoint with $T_0$.
These numbers $\lambda_{T_1}^{T_0}$ are called \emph{intersection numbers}.
For convenience, $\lambda_{T_1}^{\emptyset}$ and $\lambda_{\emptyset}^{T_0}$ are also written as $\lambda_{T_1}$ and $\lambda^{T_0}$ respectively.
The next theorem will be useful in the sequel  \cite[Theorem 9.7]{Sti08}.
\begin{theorem}\label{thm:t to 1}
Let $(\mathcal P, \mathcal B)$ be  a $t$-$(\nu,k,\lambda)$ design. Let $T_0, T_1 \subseteq \mathcal P$, where
$T_0 \cap T_1=\emptyset$, $\# T_0 =t_0$, $\# T_1=t_1$, and $t_0+t_1\le t$.
Then the intersection numbers $\lambda^{T_0}_{T_1}$ are independent of the specific choice of the elements in $T_0$ and $T_1$, and depend only on
$t_0$ and $t_1$.
Specifically,
$$\lambda^{T_0}_{T_1}= \lambda(t_0, t_1), $$
where $\lambda(t_0, t_1)=\frac{\binom{\nu-t_0-t_1}{k-t_1}}{\binom{\nu-t}{k-t}} \lambda$.
\end{theorem}

\subsection{Differentially $\delta$-uniform functions}\label{sec-duf}

Let $F$ be a vectorial Boolean function from $\mathrm{GF}(2^n)$ to $\mathrm{GF}(2^m)$.
If we use the function $F$ in a S-box of a cryptosystem, the efficiency of differential
cryptanalysis is measured by the maximum of the cardinality of the set of elements $x$  in $ \mathrm{GF}(2^n)$ such that
\begin{align*}
F(x+a)+F(x)=b,
\end{align*}
 where $a\in \mathrm{GF}(2^n)^*$ and $b\in \mathrm{GF}(2^m)$.
 The function $F$ is called a \emph{differentially $\delta$-uniform function} if
 \begin{align*}
 \max_{a\in \mathrm{GF}(2^n)^*, b \in \mathrm{GF}(2^m)} \delta(a,b)=\delta,
 \end{align*}
 where $\delta(a,b)=\# \{x\in \mathrm{GF}(2^n) : F(x+a)+F(x)=b\}$. The function
 $F$ is said to be \emph{differentially two-valued} if $\# \{\delta(a,b): a\in \mathrm{GF}(2^n)^*, b \in \mathrm{GF}(2^m) \}=2$.
 The following result can be found in \cite{BCC10}.

 \begin{proposition}\label{prop:delta=2}
 Let $F$ be a differentially $\delta$-uniform function from $\mathrm{GF}(2^n)$ to itself. Assume that $F$ is differentially two-valued.
 Then $\delta=2^s$ for some $s$, where $1\le s \le n$.
 \end{proposition}

Due to Proposition \ref{prop:delta=2}, we say that $F$ is \emph{differentially two-valued}
with $\{0, 2^s\}$
if
$$
\{\delta(a,b): a\in \mathrm{GF}(2^n)^*, b \in \mathrm{GF}(2^m) \} = \{0, 2^s\}.
$$
Results about differentially two-valued functions could be found in \cite{CP,CP19}.
When $n=m$, differentially $2$-uniform functions are also called \emph{almost perfect nonlinear} (APN) functions.

For any function $F$ from $\mathrm{GF}(2^n)$ to itself, the \emph{Walsh transform} of $F$ at $(\lambda, \mu)\in \mathrm{GF}(2^n)^* \times \mathrm{GF}(2^n)$  is defined as
\begin{align*}
\mathcal W_F(\lambda, \mu)=\sum_{x\in \mathrm{GF}(2^n)} (-1)^{ \mathrm{Tr}_{2^n/2}\left( \lambda F(x)+\mu x\right) },
\end{align*}
where $ \mathrm{Tr}_{2^n/2}(\cdot)$   is the absolute trace function from $\mathrm{GF}(2^n)$ to $\mathrm{GF}(2)$.
$\mathcal W_F(\lambda, \mu)$ are also called the  \emph{Walsh coefficients} of $F$.
 The \emph{component functions} of $F$ are the Boolean functions $\mathrm{Tr}(\lambda F(x))$, where $\lambda \in \mathrm{GF}(2^n)$.
A component function $\mathrm{Tr}(\lambda F(x))$  is said to be \emph{bent} if  $\mathcal W_F(\lambda, \mu)= \pm 2^{\frac{n}{2}}$, for all $\mu \in  \mathrm{GF}(2^n)$.
In this case, $\mathrm{Tr}(\lambda F(x))$ is also called a \emph{bent component} of $F$.

\section{Shortened and punctured codes of linear codes supporting $t$-designs}\label{sec:shpu}

Linear codes supporting $t$-designs usually have special properties \cite{Dingbk18}.
The automorphism group of
such code may be $t$-transitive or $t$-homogeneous. Such code may satisfy the conditions in
the Assmus-Mattson Theorems. Such code could be distance-optimal or dimension-optimal. In general, linear codes that support a $t$-design should have a certain kind of regularity.
Hence, one would expect that some punctured and shortened codes of such linear code would be
also attractive in certain sense. By puncturing or shortening such code, one may obtain
linear codes with different parameters and interesting properties.
This is one of the motivations of studying the punctured
and shortened codes of linear codes supporting $t$-designs. A more important motivation is for
developing a characterisation of $t$-designs supported by linear codes in Section \ref{sec:chde}.

In this section, we will first develop some general theory for some shortened and punctured
codes of linear codes supporting $t$-designs, and will then use the general theory
to determine the parameters and weight distributions of some shortened and punctured
codes of two families of binary linear codes supporting $2$-designs.

\subsection{General results for shortened and punctured codes of linear codes supporting
$t$-designs}\label{sec-pscode1}

In this subsection, we establish general results about shortened and punctured codes of linear codes supporting $t$-designs.

Recall that the binomial coefficient $\binom{a}{b}$ equals $0$ when $a<b$ or $b<0$.
Let $\mathcal W_i(\mathcal C)$  denote the set of codewords of weight $i$ in a code $\mathcal  C$ and $A_i(\mathcal{C})$ be the number of elements of $\mathcal W_i(\mathcal C)$.
We first give some results on parameters  and the weight distributions of shortened codes and punctured codes of linear codes supporting $t$-designs.
\begin{lemma}\label{lem:P:k:k+t}
Let $\mathcal C$ be a  linear code of length $\nu$ and minimum distance $d$ over $\mathrm{GF}(q)$ and  $d^{\perp}$  the minimum distance of $\mathcal  C^{\perp}$.
 Let  $t$ and $k$ be two positive integers with  $0< t <\min \{d, d^{\perp}\}$ and $1\le k\le  \nu-t$.
Let $T$ be  a  set of $t$ coordinate positions in $\mathcal  C$.
Suppose that $\left ( \mathcal P(\mathcal C) , \mathcal B_i(\mathcal C) \right )$ is a $t$-design for all $i$ with $k\le i \le k+t$.
Then
$$A_k(\mathcal C^T) =\sum_{i=0}^t  \frac{\binom{\nu-t}{k} \binom{k+i}{t} \binom{t}{i} }{\binom{\nu-t}{k-t+i} \binom{\nu}{t}} A_{k+i}(\mathcal C).$$
\end{lemma}

\begin{proof}
Let $\pi^{T}$ be the map from $\mathcal C$ to $\mathcal C^{T}$ defined as
\begin{align*}
\pi^{T}: \mathcal C & \longrightarrow \mathcal C^T,\\
(c_i)_{i \in \mathcal P(\mathcal C)} & \longmapsto (c_i)_{i\in \mathcal P(\mathcal C) \setminus T}.
\end{align*}
By Lemma \ref{lem:C-S-P}, $\pi^T$ is a one-to-one linear transformation.
Then
\begin{align*}
A_k(\mathcal C^T)= \sum_{t_1=0}^t \sum_{T_1 \subseteq T, \# T_1=t_1} \mu_{T_1}(\mathcal W_{k+t_1}(\mathcal C)),
\end{align*}
where $\mu_{T_1}(\mathcal W_{k+t_1}(\mathcal C))$
 is equal to the number of codewords in $\mathcal W_{k+t_1}(\mathcal C)$ that satisfy the conditions $c_i=0$ if $i \in T \setminus T_1 $
 and $c_i\neq 0$ if $i \in  T_1 $.
 Note that $\left ( \mathcal P(\mathcal C) , \mathcal B_{k+t_1}(\mathcal C) \right )$ is a $t$-$(\nu, k+t_1, \lambda)$
 design with $\frac{1}{q-1} A_{k+t_1}(\mathcal C)$ blocks, where $\lambda=\frac{\binom{k+t_1}{t}}{\binom{\nu}{t}} \frac{1}{q-1}  A_{k+t_1}(\mathcal C)$.
 Let $ \lambda_{T_1}^{T\setminus T_1}$ be the intersection number
 of the $t$-design $\left ( \mathcal P(\mathcal C) , \mathcal B_{k+t_1}(\mathcal C) \right )$.
 By Theorem \ref{thm:t to 1}, one has
 \begin{align*}
 \mu_{T_1}(\mathcal W_{k+t_1}(\mathcal C))=&(q-1) \lambda_{T_1}^{T\setminus T_1}\\
 =& (q-1) \frac{\binom{\nu-t}{k+t_1-t_1}}{ \binom{\nu-t}{k+t_1-t}} \lambda\\
 =& \frac{\binom{\nu-t}{k}\binom{k+t_1}{t}}{ \binom{\nu-t}{k-t+t_1} \binom{\nu}{t}} A_{k+t_1}(\mathcal C).
 \end{align*}
It then follows that
$$
A_k(\mathcal C^T)= \sum_{t_1=0}^t \binom{t}{t_1}  \frac{\binom{\nu-t}{k}\binom{k+t_1}{t}}{ \binom{\nu-t}{k-t+t_1} \binom{\nu}{t}} A_{k+t_1}(\mathcal C).
$$
\end{proof}

\begin{theorem}\label{thm:sct-code}
Let $\mathcal C$ be a $[\nu, m, d]$ linear code  over $\mathrm{GF}(q)$ and  $d^{\perp}$  the minimum distance of $\mathcal  C^{\perp}$.
 Let  $t$ be a positive integer with  $0< t <\min \{d, d^{\perp}\}$.
Let $T$ be  a  set of $t$ coordinate positions in $\mathcal  C$.
Suppose that $\left ( \mathcal P(\mathcal C) , \mathcal B_i(\mathcal C) \right )$ is a $t$-design for any $i$ with $d \le i \le \nu-t$.
Then the shortened code $\mathcal C_T$ is a linear code of length $\nu-t$ and dimension $m-t$. The weight distribution
$\left ( A_k(\mathcal C_T) \right )_{k=0}^{\nu-t}$ of $\mathcal C_T$ is independent of the specific choice of the elements
in $T$. Specifically,
$$A_k(\mathcal C_T) =\frac{ \binom{k}{t} \binom{\nu-t}{k}}{ \binom{\nu }{t} \binom{\nu-t}{k-t}}A_k(\mathcal C).$$
\end{theorem}

\begin{proof}
Let $\mathcal C(T)=\{(c_i)_{i \in \mathcal C} \in \mathcal C: c_i=0 \text{ for any } i \in T \}$.
Let $\pi_{T}$ be the map from $\mathcal C(T)$ to $\mathcal C_{T}$ defined as
\begin{align*}
\pi_{T}: \mathcal C(T) & \longrightarrow \mathcal C_T,\\
(c_i)_{i \in \mathcal P(\mathcal C)} & \longmapsto (c_i)_{i\in \mathcal P(\mathcal C) \setminus T}.
\end{align*}
By the definition of $\mathcal C(T)$ and $\mathcal C_{T}$, the map $\pi_T$ is a  one-to-one linear transformation.
Then
\begin{align*}
A_k(\mathcal C_T)=    \mu^{T}(\mathcal W_{k}(\mathcal C)),
\end{align*}
where $\mu^{T}(\mathcal W_{k}(\mathcal C))$
 is equal to the number of codewords in $\mathcal W_{k}(\mathcal C)$ that satisfy the conditions $c_i=0$ if $i \in T $.
 Note that $\left ( \mathcal P(\mathcal C) , \mathcal B_{k}(\mathcal C) \right )$ is a $t$-$(\nu, k, \lambda)$
 design with $\frac{1}{q-1} A_{k}(\mathcal C)$ blocks, where $\lambda=\frac{\binom{k}{t}}{\binom{\nu}{t}} \frac{1}{q-1}  A_{k}(\mathcal C)$.
 Let $ \lambda^{T}$ be the intersection number
 of the $t$-design $\left ( \mathcal P(\mathcal C) , \mathcal B_{k}(\mathcal C) \right )$. By  Theorem \ref{thm:t to 1}, one has
 \begin{align*}
 \mu^{T}(\mathcal W_{k}(\mathcal C))=&(q-1) \lambda^{T}\\
 =& (q-1) \frac{\binom{\nu-t}{k}}{\binom{\nu-t}{k-t}} \lambda\\
 =& \frac{ \binom{k}{t} \binom{\nu-t}{k}}{ \binom{\nu }{t} \binom{\nu-t}{k-t}}    A_{k}(\mathcal C).
 \end{align*}
The desired conclusion then follows from
$
A_k(\mathcal C_T)=    \mu^{T}(\mathcal W_{k}(\mathcal C))
$
and Lemma \ref{lem:C-S-P}.
\end{proof}

\begin{theorem}\label{thm:pct-code}
Let $\mathcal C$ be a $[\nu, m, d]$ linear code over $\mathrm{GF}(q)$ and  $d^{\perp}$  the minimum distance of $\mathcal  C^{\perp}$.
 Let  $t$ be a positive integer with  $0< t <d^{\perp}$.
Let $T$ be  a  set of $t$ coordinate positions in $\mathcal  C$.
Suppose that $\left ( \mathcal P(\mathcal C) , \mathcal B_i(\mathcal C) \right )$ is a $t$-design for any $i$ with $d \le i \le \nu$.
Then the punctured code $\mathcal C^T$ is a linear code of length $\nu-t$ and dimension $m$. The weight distribution
$\left ( A_k(\mathcal C^T) \right )_{k=0}^{\nu-t}$ of $\mathcal C^T$ is independent of the specific choice of the elements
in $T$. Specifically,
$$A_k(\mathcal C^T) =\sum_{i=0}^t  \frac{\binom{\nu-t}{k} \binom{k+i}{t} \binom{t}{i} }{\binom{\nu-t}{k-t+i} \binom{\nu}{t}} A_{k+i}(\mathcal C).$$
\end{theorem}

\begin{proof}
The desired results follow from Lemmas \ref{lem:C-S-P} and \ref{lem:P:k:k+t}.
\end{proof}

Theorems \ref{thm:sct-code} and \ref{thm:pct-code} settle the parameters and
weight distribution of the shortened code $\mathcal C_T$ and punctured code
$\mathcal C^T$ of a code $\mathcal C$ supporting $t$-designs, respectively.
In general it could be very hard to determine the weight distribution of a
shortened or punctured code of a linear code.

\subsection{Punctured and shortened codes of a family of binary codes}\label{sec-pscode2}

In this subsection, we determine the parameters and weight distributions of
some punctured and shortened codes of a family of binary
linear codes constructed from bent Boolean functions. As will be demonstrated
shortly, the shortened and punctured codes are quite interesting.

Let $f$ be a bent function from $\mathrm{GF}(2^n)$ to $\mathrm{GF}(2)$, and let $D_f=\{d_0, d_1, \ldots, d_{\nu_f-1}\} \subseteq \mathrm{GF}(2^n)$ be the support of $f$.
Define a binary code of length $\nu_f$ by
\begin{align*}
\mathcal C(D_f)=\{\left ( \mathrm{Tr}_{2^n/2}(x d_0)+y, \ldots,   \mathrm{Tr}_{2^n/2}(x d_{\nu_f-1})+y \right ): x\in \mathrm{GF}(2^n), y\in \mathrm{GF}(2) \}.
\end{align*}
The following theorem on parameters of
 $\mathcal C(D_f)$ was proved in \cite[Theorems 14.13 and 14.15]{Dingbk18}.
\begin{theorem}\label{thm-C(supp-bent)}
Let $f$ be a bent function from $\mathrm{GF}(2^n)$ to $\mathrm{GF}(2)$, where $n\ge 6$ and is even. Then $\mathcal C(D_f)$ is a  $[\nu_f, n+1, (\nu_f-2^{\frac{n-2}{2}})/2]$
three-weight binary code with the weight distribution in Table \ref{table:supp-bent} and it  holds $2$-designs.
The dual code $\mathcal C(D_f)^\perp$ has minimum distance $4$.

\begin{table}[htbp]
\centering
\caption{The weight distribution of the code  $\mathcal C(D_f)$ of Theorem \ref{thm-C(supp-bent)}}
\label{table:supp-bent}
\begin{tabular}{|c|c|}
  \hline
  Weight & Multiplicity \\
  \hline
  $0$& $1$\\
\hline
$\frac{\nu_f}{2}-2^{\frac{n-4}{2}}$ & $2^n-1$\\
\hline
$\frac{\nu_f}{2}+2^{\frac{n-4}{2}}$ & $2^n-1$\\
 \hline
 $\nu_f$ & $1$\\
\hline
\end{tabular}
\end{table}

\end{theorem}

Taking $T=\{t_1\}$, we have the parameters and the weight distribution of the shortened code
$\mathcal C(D_f)_{\{t_1\}}$ of $\mathcal C(D_f)$ in the following theorem.
\begin{theorem}\label{thm:sc1-bent}
Let $t_1$ be an integer with $0\le t_1 <\nu_f$. Let $f$ be a bent function from $\mathrm{GF}(2^n)$ to $\mathrm{GF}(2)$, where $n\ge 6$ and is even.
Then, the shortened code $\mathcal C(D_f)_{\{t_1\}}$ is a two-weight binary linear code of length $\nu_f-1$ and dimension $n$, and has the weight distribution in Table \ref{table:sc1-bent}.
\begin{table}[htbp]
\centering
\caption{The weight distribution of the shortened code $\mathcal C(D_f)_{\{t_1\}}$ of Theorem \ref{thm:sc1-bent}}
\label{table:sc1-bent}
\begin{tabular}{|c|c|}
  \hline
  Weight & Multiplicity \\
  \hline
  $0$& $1$\\
\hline
$\frac{\nu_f}{2}-2^{\frac{n-4}{2}}$ & $ \frac{\nu_f+2^{\frac{n-2}{2}}}{2 \nu_f}(2^n-1)$\\
\hline
$\frac{\nu_f}{2}+2^{\frac{n-4}{2}}$ & $ \frac{\nu_f-2^{\frac{n-2}{2}}}{2 \nu_f}(2^n-1)$\\
\hline
\end{tabular}
\end{table}

\end{theorem}

\begin{proof}
By Theorem \ref{thm:sct-code},
\begin{align*}
A_k(\mathcal C (D_f)_{\{t_1\}}) =\frac{\nu_f-k}{\nu_f} A_k\left (\mathcal C (D_f)\right ).
\end{align*}
The desired results follow from Theorem \ref{thm-C(supp-bent)}.
\end{proof}

Taking $T=\{t_1,t_2\}$, we have the parameters and the weight distribution of the shortened code
$\mathcal C(D_f)_{\{t_1,t_2\}}$ of $\mathcal C(D_f)$ in the following theorem.

\begin{theorem}\label{thm:sc2-bent}
Let $t_1$ and $t_2$ be  integers with $0\le t_1< t_2 <\nu_f$. Let $f$ be a bent function from
$\mathrm{GF}(2^n)$ to $\mathrm{GF}(2)$, where $n\ge 6$ and is even.
Then, the shortened code $\mathcal C(D_f)_{\{t_1, t_2\}}$ is a two-weight binary linear code of length $\nu_f-2$ and dimension $n-1$, and
has the weight distribution in Table \ref{table:sc2-bent}.
\begin{table}[htbp]
\centering
\caption{The weight distribution of the shortened code $\mathcal C(D_f)_{\{t_1, t_2\}}$ of Theorem \ref{thm:sc2-bent}}
\label{table:sc2-bent}
\begin{tabular}{|c|c|}
  \hline
  Weight & Multiplicity \\
  \hline
  $0$& $1$\\
\hline
$\frac{\nu_f}{2}-2^{\frac{n-4}{2}}$ & $ \frac{\left (\nu_f+2^{\frac{n-2}{2}} \right ) \left (\nu_f+2^{\frac{n-2}{2}}-2 \right )}{4 \nu_f(\nu_f-1)}(2^n-1)$\\
\hline
$\frac{\nu_f}{2}+2^{\frac{n-4}{2}}$ & $  \frac{\left (\nu_f-2^{\frac{n-2}{2}} \right ) \left (\nu_f-2^{\frac{n-2}{2}}-2 \right )}{4 \nu_f(\nu_f-1)}(2^n-1)$\\
\hline
\end{tabular}
\end{table}

\end{theorem}

\begin{proof}
By Theorem \ref{thm:sct-code},
\begin{align*}
A_k(\mathcal C (D_f)_{\{t_1, t_2\}}) =\frac{(\nu_f-k)(\nu_f-k-1)}{\nu_f(\nu_f-1)} A_k\left (\mathcal C (D_f)\right ).
\end{align*}
The desired results follow from Theorem \ref{thm-C(supp-bent)}.
\end{proof}

Taking $T=\{t_1\}$, we have the parameters and the weight distribution of the punctured code
$\mathcal C(D_f)^{\{t_1\}}$ of $\mathcal C(D_f)$ in the following theorem.

\begin{theorem}\label{thm:pc1-bent}
Let $t_1$ be an integer with $0\le t_1 <\nu_f$. Let $f$ be a bent function from $\mathrm{GF}(2^n)$ to $\mathrm{GF}(2)$, where $n\ge 6$ and is even.
Then, the punctured code $\mathcal C(D_f)^{\{t_1\}}$ is a five-weight binary linear code of length $\nu_f-1$ and dimension $n+1$, and has the weight distribution in Table \ref{table:pc1-bent}.
\begin{table}[htbp]
\centering
\caption{The weight distribution of the punctured  code $\mathcal C(D_f)^{\{t_1\}}$ of Theorem \ref{thm:pc1-bent}}
\label{table:pc1-bent}
\begin{tabular}{|c|c|}
  \hline
  Weight & Multiplicity \\
  \hline
  $0$& $1$\\
  \hline
  $\frac{\nu_f}{2}-2^{\frac{n-4}{2}}-1$ & $ \frac{\nu_f-2^{\frac{n-2}{2}}}{2 \nu_f}(2^n-1)$\\
\hline
$\frac{\nu_f}{2}-2^{\frac{n-4}{2}}$ & $ \frac{\nu_f+2^{\frac{n-2}{2}}}{2 \nu_f}(2^n-1)$\\
\hline
$\frac{\nu_f}{2}+2^{\frac{n-4}{2}}-1$ & $ \frac{\nu_f+2^{\frac{n-2}{2}}}{2 \nu_f}(2^n-1)$\\
\hline
$\frac{\nu_f}{2}+2^{\frac{n-4}{2}}$ & $ \frac{\nu_f-2^{\frac{n-2}{2}}}{2 \nu_f}(2^n-1)$\\
\hline
$\nu_f-1$ & $1$\\
\hline
\end{tabular}
\end{table}

\end{theorem}

\begin{proof}
By Theorem \ref{thm:pct-code}, for $k= \frac{\nu_f}{2}\pm 2^{\frac{n-4}{2}}$, one has
\begin{align*}
A_k(\mathcal C (D_f)^{\{t_1\}}) =\frac{\nu_f-k}{\nu} A_k\left (\mathcal C (D_f)\right ),
\end{align*}
and
\begin{align*}
A_{k-1}(\mathcal C (D_f)^{\{t_1\}}) =\frac{k}{\nu_f} A_k\left (\mathcal C (D_f)\right ).
\end{align*}
The desired results follow from Theorem \ref{thm-C(supp-bent)}.
\end{proof}

Taking $T=\{t_1,t_2\}$, we have the parameters and the weight distribution of the punctured code
$\mathcal C(D_f)^{\{t_1,t_2\}}$ of $\mathcal C(D_f)$ in the following theorem.

\begin{theorem}\label{thm:pc2-bent}
Let $t_1, t_2$ be  integers with $0\le t_1<t_2 <\nu_f$. Let $f$ be a bent function from
$\mathrm{GF}(2^n)$ to $\mathrm{GF}(2)$, where $n\ge 6$ and is even.
Then, the punctured code $\mathcal C(D_f)^{\{t_1, t_2\}}$ is a seven-weight binary linear code of length $\nu_f-2$ and dimension $n+1$, and has the weight distribution in Table \ref{table:pc2-bent}.
\begin{table}[htbp]
\centering
\caption{The weight distribution of the punctured  code $\mathcal C(D_f)^{\{t_1,t_2\}}$ of Theorem \ref{thm:pc2-bent}}
\label{table:pc2-bent}
\begin{tabular}{|c|c|}
  \hline
  Weight & Multiplicity \\
  \hline
  $0$& $1$\\
  \hline
   $\frac{\nu_f}{2}-2^{\frac{n-4}{2}}-2$ & $  \frac{\left (\nu_f-2^{\frac{n-2}{2}} \right ) \left (\nu_f-2^{\frac{n-2}{2}}-2 \right )}{4 \nu_f(\nu_f-1)}(2^n-1) $\\
  \hline
  $\frac{\nu_f}{2}-2^{\frac{n-4}{2}}-1$ & $ \frac{\nu_f^2-2^{n-2}}{2 \nu_f(\nu_f-1)}(2^n-1)$\\
\hline
$\frac{\nu_f}{2}-2^{\frac{n-4}{2}}$ & $  \frac{\left (\nu_f+2^{\frac{n-2}{2}} \right ) \left (\nu_f+2^{\frac{n-2}{2}}-2 \right )}{4 \nu_f(\nu_f-1)}(2^n-1)$\\
\hline
$\frac{\nu_f}{2}+2^{\frac{n-4}{2}}-2$ & $ \frac{\left (\nu_f+2^{\frac{n-2}{2}} \right ) \left (\nu_f+2^{\frac{n-2}{2}}-2 \right )}{4 \nu_f(\nu_f-1)}(2^n-1) $\\
\hline
$\frac{\nu_f}{2}+2^{\frac{n-4}{2}}-1$ & $ \frac{\nu_f^2-2^{n-2}}{2 \nu_f(\nu_f-1)}(2^n-1)$\\
\hline
$\frac{\nu_f}{2}+2^{\frac{n-4}{2}}$ & $  \frac{\left (\nu_f-2^{\frac{n-2}{2}} \right ) \left (\nu_f-2^{\frac{n-2}{2}}-2 \right )}{4 \nu_f(\nu_f-1)}(2^n-1) $\\
\hline
$\nu_f-2$ & $1$\\
\hline
\end{tabular}
\end{table}

\end{theorem}

\begin{proof}
By Theorem \ref{thm:pct-code}, for $k= \frac{\nu_f}{2}\pm 2^{\frac{n-4}{2}}$, one has
\begin{align*}
A_k(\mathcal C (D_f)^{\{t_1,t_2\}}) =A_k\left ( \mathcal C (D_f)_{\{t_1, t_2\}}  \right ),
\end{align*}
\begin{align*}
A_{k-1}(\mathcal C (D_f)^{\{t_1,t_2\}}) =\frac{2k(\nu_f-k)}{\nu_f(\nu_f-1)}   A_k\left (\mathcal C (D_f)\right )
\end{align*}
and
\begin{align*}
A_{k-2}(\mathcal C (D_f)^{\{t_1,t_2\}}) =\frac{k(k-1)}{\nu_f(\nu_f-1)}   A_k\left (\mathcal C (D_f)\right ).
\end{align*}
The desired results follow from Theorem \ref{thm-C(supp-bent)} and Theorem \ref{thm:sc2-bent}.
\end{proof}

\begin{example}
Let $\mathrm{GF}(2^6)=\mathrm{GF}(2)[u]/\left(u^6+u^4+u^3+u+1\right)$ and $\alpha\in \mathrm{GF}(2^6)$ such that $\alpha^6+\alpha^4+\alpha^3+\alpha+1=0$.
Then $\alpha$ is a primitive element of $ \mathrm{GF}(2^6)$ and $f(x)=\mathrm{Tr}_{2^6/2}(\alpha  x^3)$ is a bent function on $ \mathrm{GF}(2^6)$ with
$\nu_f=\# D_f=36$.
$\mathcal{C}(D_f)$ is a $[36,7,16]$ linear code with  weight enumerator $1+63z^{16}+ 63z^{20}+z^{36}$.

Let $t_1$ be an integer with $0\le t_1\le 35$. Then the shortened code $\mathcal C(D_f)_{\{t_1\}}$  has parameters $[35,6,16]$  and weight enumerator $1+35z^{16}+ 28z^{20}$.
The punctured code $\mathcal C(D_f)^{\{t_1\}}$  has parameters $[35,7,15]$  and weight enumerator $1+28z^{15}+ 35z^{16}+35z^{19}+28z^{20}+z^{35}$. The code  $\mathcal C(D_f)_{\{t_1\}}$ is optimal and the code  $\mathcal C(D_f)^{\{t_1\}}$ is almost optimal with respect to  the Griesmer bound.

Let $t_1$ and $t_2$ be two integers with $0\le t_1<t_2\le 35$. Then the shortened code $\mathcal C(D_f)_{\{t_1, t_2\}}$  has parameters $[34, 5, 16]$  and weight enumerator
$1+19z^{16}+ 12z^{20}$.
The punctured code $\mathcal C(D_f)^{\{t_1, t_2\}}$  has parameters $[34, 7, 14]$  and weight enumerator $1+12z^{14}+32z^{15}+ 19z^{16}+19z^{18}+32z^{19}+12z^{20}+z^{34}$.
The code
 $\mathcal C(D_f)_{\{t_1,t_2\}}$ is optimal and the code $\mathcal C(D_f)^{\{t_1, t_2\}}$ is almost optimal with respect to  the Griesmer bound.
\end{example}

\subsection{Punctured and shortened codes of another family of binary codes}\label{sec-pscode3}

In this subsection, we settle the parameters and weight distributions of
some punctured and shortened codes of another family of binary
linear codes constructed from bent vectorial Boolean functions.
It will be shown that the shortened and punctured codes are
interesting.

Let $F(x)$ be a vectorial  function   from  $\mathrm{GF}(2^{n})$ to $\mathrm{GF}(2^{\ell})$.
Let $\mathcal C(F)$ be the binary code of length $2^{n}$ defined by
\begin{align}\label{eq:d-c(vb)}
\mathcal C(F)=\left  \{\left ( c_{a,b,c}(x)    \right )_{x\in \mathrm{GF}(2^{n})}:  (a,b,c)\in \mathrm{GF}(2^l)\times \mathrm{GF}(2^{n}) \times \mathrm{GF}(2) \right \},
\end{align}
where $c_{a,b,c}(x)=\mathrm{Tr}_{2^{\ell}/2}(a F(x))+ \mathrm{Tr}_{2^{n}/2}(b x) +c $.

The following was proved in \cite[Theorem 5]{DMT18}.
\begin{theorem}\label{thm-C(bent)}
Let $F$ be a bent vectorial function from $\mathrm{GF}(2^{2m})$ to $\mathrm{GF}(2^\ell)$, where $m\ge 3$. Then $\mathcal C(F)$ is a  $[2^{2m}, 2m+\ell+1, 2^{2m-1}-2^{m-1}]$
four-weight binary code with the weight distribution in Table \ref{table:supp-bentcode}.
The dual code $\mathcal C(F)^\perp$ has minimum distance $4$.

\begin{table}[htbp]
\centering
\caption{The weight distribution of the code $\mathcal C(F)$ of Theorem \ref{thm-C(bent)}}
\label{table:supp-bentcode}
\begin{tabular}{|c|c|}
  \hline
  Weight & Multiplicity \\
  \hline
  $0$& $1$\\
\hline
$2^{2m-1}-2^{m-1}$ & $(2^{l}-1)2^{2m}$\\
\hline
$2^{2m-1}$ & $2(2^{2m}-1)$\\
\hline
$2^{2m-1}+2^{m-1}$ & $(2^{l}-1)2^{2m}$\\
 \hline
 $2^{2m}$ & $1$\\
\hline
\end{tabular}
\end{table}

\end{theorem}

Taking $T=\{t_1\}$, we have the parameters and the weight distribution of the shortened code
$\mathcal C(F)_{\{t_1\}}$ of $\mathcal C(F)$  in the following theorem.
\begin{theorem}\label{thm:sc1-vbent}
Let $t_1$, $m$ be  integers with $0\le t_1 <2^{2m}$ and $m\ge 3$. Let $F$ be a bent vectorial function from $\mathrm{GF}(2^{2m})$ to $\mathrm{GF}(2^\ell)$.
Then, the shortened code $\mathcal C(F)_{\{t_1\}}$ is a binary linear code of length $2^{2m}-1$ and dimension $2m+\ell$, and has the weight distribution in Table \ref{table:sc1-vbent}.
\begin{table}[htbp]
\centering
\caption{The weight distribution of the code $\mathcal C(F)_{\{t_1\}}$ of Theorem \ref{thm:sc1-vbent}}
\label{table:sc1-vbent}
\begin{tabular}{|c|c|}
  \hline
  Weight & Multiplicity \\
  \hline
  $0$& $1$\\
\hline
$2^{2m-1}-2^{m-1}$ & $(2^{l}-1)\left (2^{2m-1}+2^{m-1}\right )$\\
\hline
$2^{2m-1}$ & $2^{2m}-1$\\
\hline
$2^{2m-1}+2^{m-1}$ & $(2^{l}-1)\left (2^{2m-1}-2^{m-1}\right )$\\
\hline
\end{tabular}
\end{table}

\end{theorem}
\begin{proof}
By Theorem \ref{thm:sct-code},
\begin{align*}
A_k(\mathcal C (F)_{\{t_1\}}) =\frac{2^{2m}-k}{2^{2m}} A_k\left (\mathcal C (F)\right ).
\end{align*}
The desired results follow from Theorem \ref{thm-C(bent)}.
\end{proof}

Taking $T=\{t_1,t_2\}$, we have the parameters and the weight distribution of the shortened code
$\mathcal C(F)_{\{t_1,t_2\}}$ of $\mathcal C(F)$  in the following theorem.

\begin{theorem}\label{thm:sc2-vbent}
Let $t_1$, $t_2$ and $m$ be  integers with $0\le t_1< t_2 <2^{2m}$ and $m\ge 3$. Let $F$ be a bent vectorial function from $\mathrm{GF}(2^{2m})$ to $\mathrm{GF}(2^\ell)$.
Then, the shortened code $\mathcal C(F)_{\{t_1,t_2\}}$ is a binary linear code of length $2^{2m}-2$ and dimension $2m+\ell-1$, and has the weight distribution in Table \ref{table:sc2-vbent}.
\begin{table}[htbp]
\centering
\caption{The weight distribution of the code $\mathcal C(F)_{\{t_1,t_2\}}$ of Theorem \ref{thm:sc2-vbent}}
\label{table:sc2-vbent}
\begin{tabular}{|c|c|}
  \hline
  Weight & Multiplicity \\
  \hline
  $0$& $1$\\
\hline
$2^{2m-1}-2^{m-1}$ & $(2^{l}-1)2^{m-2} \left ( 2^m+2 \right )$\\
\hline
$2^{2m-1}$ & $2^{2m-1}-1$\\
\hline
$2^{2m-1}+2^{m-1}$ & $(2^{l}-1)2^{m-2} \left ( 2^m-2 \right )$\\
\hline
\end{tabular}
\end{table}

\end{theorem}
\begin{proof}
By Theorem \ref{thm:sct-code},
\begin{align*}
A_k(\mathcal C (F)_{\{t_1, t_2\}}) =\frac{(2^{2m}-k)(2^{2m}-k-1)}{2^{2m}(2^{2m}-1)} A_k\left (\mathcal C (F)\right ).
\end{align*}
The desired results follow from Theorem \ref{thm-C(bent)}.
\end{proof}

Taking $T=\{t_1\}$, we have the parameters and the weight distribution of the punctured code
$\mathcal C(F)^{\{t_1\}}$ of $\mathcal C(F)$  in the following theorem.

\begin{theorem}\label{thm:pc1-vbent}
Let $t_1$ and $m$ be  integers with $0\le t_1 <2^{2m}$ and $m\ge 3$. Let $F$ be a bent vectorial function from $\mathrm{GF}(2^{2m})$ to $\mathrm{GF}(2^\ell)$.
Then, the punctured code $\mathcal C(F)^{\{t_1\}}$ is a binary linear code of length $2^{2m}-1$ and dimension $2m+\ell+1$, and has the weight distribution in Table \ref{table:pc1-vbent}.
\begin{table}[htbp]
\centering
\caption{The weight distribution of the code $\mathcal C(F)^{\{t_1\}}$ of Theorem \ref{thm:pc1-vbent}}
\label{table:pc1-vbent}
\begin{tabular}{|c|c|}
  \hline
  Weight & Multiplicity \\
  \hline
  $0$& $1$\\
\hline
$2^{2m-1}-2^{m-1}-1$ & $(2^{l}-1)\left (2^{2m-1}-2^{m-1}\right )$\\
\hline
$2^{2m-1}-2^{m-1}$ & $(2^{l}-1)\left (2^{2m-1}+2^{m-1}\right )$\\
\hline
$2^{2m-1}-1$ & $2^{2m}-1$\\
\hline
$2^{2m-1}$ & $2^{2m}-1$\\
\hline
$2^{2m-1}+2^{m-1}-1$ & $(2^{l}-1)\left (2^{2m-1}+2^{m-1}\right )$\\
\hline
$2^{2m-1}+2^{m-1}$ & $(2^{l}-1)\left (2^{2m-1}-2^{m-1}\right )$\\
\hline
$2^m-1$ & $1$\\
\hline
\end{tabular}
\end{table}

\end{theorem}

\begin{proof}
By Theorem \ref{thm:pct-code}, for $k\in \left \{2^{2m-1}-2^{m-1},  2^{2m-1}, 2^{2m-1}+2^{m-1} \right \}$, one has
\begin{align*}
A_k(\mathcal C (F)^{\{t_1\}}) =\frac{2^{2m}-k}{2^{2m}} A_k\left (\mathcal C (F)\right ),
\end{align*}
and
\begin{align*}
A_{k-1}(\mathcal C (F)^{\{t_1\}}) =\frac{k}{2^{2m}} A_k\left (\mathcal C (F)\right ).
\end{align*}
The desired results follow from Theorem \ref{thm-C(bent)}.
\end{proof}

Taking $T=\{t_1,t_2\}$, we have the parameters and the weight distribution of the punctured code
$\mathcal C(F)^{\{t_1,t_2\}}$ of $\mathcal C(F)$  in the following theorem.

\begin{theorem}\label{thm:pc2-vbent}
Let $t_1$, $t_2$ and $m$ be  integers with $0\le t_1< t_2 <2^{2m}$ and $m\ge 3$. Let $F$ be a bent vectorial function from $\mathrm{GF}(2^{2m})$ to $\mathrm{GF}(2^\ell)$.
Then, the punctured code $\mathcal C(F)^{\{t_1,t_2\}}$ is a binary linear code of length $2^{2m}-2$ and dimension $2m+\ell+1$, and has the weight distribution in Table \ref{table:pc2-vbent}.
\begin{table}[htbp]
\centering
\caption{The weight distribution of the code $\mathcal C(F)^{\{t_1,t_2\}}$ of Theorem \ref{thm:pc2-vbent}}
\label{table:pc2-vbent}
\begin{tabular}{|c|c|}
  \hline
  Weight & Multiplicity \\
  \hline
  $0$& $1$\\
\hline
$2^{2m-1}-2^{m-1}-2$ & $2^{m-2}(2^{l}-1)(2^{m}-2)$\\
\hline
$2^{2m-1}-2^{m-1}-1$ & $2^{2m-1}(2^{l}-1)$\\
\hline
$2^{2m-1}-2^{m-1}$ & $(2^{l}-1)2^{m-2} \left ( 2^m+2 \right )$\\
\hline
$2^{2m-1}-2$ & $2^{2m-1}-1$\\
\hline
$2^{2m-1}-1$ & $2^{2m}$\\
\hline
$2^{2m-1}$ & $2^{2m-1}-1$\\
\hline
$2^{2m-1}+2^{m-1}-2$ & $2^{m-2}(2^{l}-1)(2^{m}+2)$\\
\hline
$2^{2m-1}+2^{m-1}-1$ & $2^{2m-1}(2^{l}-1)$\\
\hline
$2^{2m-1}+2^{m-1}$ & $(2^{l}-1)2^{m-2} \left ( 2^m-2 \right )$\\
\hline
$2^m-2$ & $1$\\
\hline
\end{tabular}
\end{table}

\end{theorem}

\begin{proof}
By Theorem \ref{thm:pct-code}, for $k\in \left \{   2^{2m-1}-2^{m-1}, 2^{2m-1}, 2^{2m-1}+2^{m-1}  \right \}$, one has
\begin{align*}
A_k(\mathcal C (F)^{\{t_1,t_2\}}) =A_k\left ( \mathcal C (F)_{\{t_1, t_2\}}  \right),
\end{align*}
\begin{align*}
A_{k-1}(\mathcal C (F)^{\{t_1,t_2\}}) =\frac{2k(2^{2m}-k)}{2^{2m}(2^{2m}-1)}   A_k\left (\mathcal C (F)\right )
\end{align*}
and
\begin{align*}
A_{k-2}(\mathcal C (F)^{\{t_1,t_2\}}) =\frac{k(k-1)}{2^{2m}(2^{2m}-1)}   A_k\left (\mathcal C (F)\right ).
\end{align*}
The desired results follow from Theorem \ref{thm-C(bent)} and Theorem \ref{thm:sc2-vbent}.
\end{proof}

\begin{example}
Let $\mathrm{GF}(2^6)=\mathrm{GF}(2)[u]/\left(u^6+u^4+u^3+u+1\right)$ and $\alpha\in \mathrm{GF}(2^6)$ such that $\alpha^6+\alpha^4+\alpha^3+\alpha+1=0$.
Then $F(x)=\mathrm{Tr}_{2^6/2^3}(\alpha  x^3)$ is a bent vectorial function from $ \mathrm{GF}(2^6)$ to $ \mathrm{GF}(2^3)$. The code
$\mathcal{C}(F)$ is a $[64, 10, 28]$ linear code with  weight enumerator $1+448z^{28}+ 126z^{32}+448z^{36}+z^{64}$.

Let $t_1$ be an integer with $0\le t_1\le 63$. Then the shortened code $\mathcal C(F)_{\{t_1\}}$  has parameters $[63, 9, 28]$  and weight enumerator
$1+252z^{28}+ 63z^{32}+196z^{36}$.
The punctured code $\mathcal C(F)^{\{t_1\}}$  has parameters $[63, 10, 27]$  and weight enumerator
$1+196z^{27}+ 252z^{28}+63z^{31}+63z^{32}+252z^{35}+196z^{36}
+z^{63}$. The code
$\mathcal C(F)_{\{t_1\}}$ is optimal with respect to  a one-step Griesmer bound, and $\mathcal C(F)^{\{t_1\}}$ has the same parameters
as the best binary linear code known in the database maintained by Markus Grassl.

Let $t_1$ and $t_2$ be two integers with $0\le t_1<t_2\le 63$. Then the shortened code $\mathcal C(F)_{\{t_1, t_2\}}$  has parameters $[62, 8, 28]$  and weight enumerator
$1+140z^{28}+ 31z^{32}+84z^{36}$.
The punctured code $\mathcal C(F)^{\{t_1, t_2\}}$  has parameters $[62, 10, 26]$  and weight enumerator
$1+84z^{26}+224z^{27}+ 140z^{28}+31z^{30}+64z^{31}+31z^{32}
+140z^{34}+224z^{35}+84z^{36}+z^{62}$. The code
$\mathcal C(F)_{\{t_1,t_2\}}$ is optimal with respect to  a one-step Griesmer bound, and $\mathcal C(F)^{\{t_1,t_2\}}$ has the same parameters
as the best binary linear code known in the database maintained by Markus Grassl.

\end{example}

\section{Characterizations of linear codes supporting $t$-designs via shortened and punctured  codes}\label{sec:chde}

In this section, we shall give a characterization of codes supporting $t$-designs in terms of their shortened and punctured codes.
Let $\mathcal P$ be a set of $\nu$ elements  and $\mathcal B$   a \textbf{multiset} of  $k$-subsets of $\mathcal P$, where  $1\le k \le \nu$.
Let $\overline {\mathcal B}=\{\{ \mathcal P \setminus B: B \in \mathcal B \}\}$.

\begin{lemma}\label{lem:bigblock}
Let $(\mathcal P, \mathcal B)$ be a $(\nu-k)$-$(\nu, k, \lambda)$ design and $t$  an integer with $1\le \nu -k \le t \le k$.
Then  $(\mathcal P, \mathcal B)$ is also a $t$-$\left (\nu, k, \binom{\nu-t}{\nu-k} \lambda / \binom{\nu-t}{k-t} \right )$ design.
\end{lemma}

\begin{proof}
Let $T$ be any $t$-subset of $\mathcal P$. It is observed that
\begin{align*}
\{\{B\in \mathcal B: T \subseteq B\}\}=\cup_{T'\subseteq \mathcal P\setminus T, \# T'=\nu-k} \{\{B\in \mathcal B: B \cup T'= \mathcal P \}\}.
\end{align*}
Then
\begin{align*}
\lambda_{T}= \sum_{T'\subseteq \mathcal P\setminus T, \# T'=\nu-k} \lambda^{T'},
\end{align*}
where $\lambda_{T}$ and $\lambda^{T'}$ are the intersection numbers of the design $(\mathcal P, \mathcal B)$.
By Theorem \ref{thm:t to 1}, one gets
\begin{align*}
\lambda_{T}= &\binom{\nu-t}{\nu-k} \lambda^{T'}\\
=& \binom{\nu-t}{\nu-k}   \frac{\binom{\nu-(\nu-k)}{k}}{\binom{\nu-t}{k-t}} \lambda\\
=&  \frac{\binom{\nu-t}{\nu-k} }{\binom{\nu-t}{k-t}}  \lambda.
\end{align*}
It completes the proof.
\end{proof}

In the case of simple designs, Lemma \ref{lem:bigblock} was known in the literature.
The conclusion of Lemma \ref{lem:bigblock} implies that a $(\nu-k)$-$(\nu, k, \lambda)$
design must be a trivial design, as every $k$-subset of the point set is a block of the
design.

\begin{lemma}\label{lem:codeTOdesign}
Let $\mathbb D=(\mathcal P, \mathcal B)$ be a $t$-$(\nu, k, \lambda)$ design with $t\le k \le \nu-t$. Then $\overline {\mathbb D}=(\mathcal P, \overline {\mathcal B})$
is a $t$-$(\nu, \nu-k, \overline {\lambda})$ design, where $\overline {\lambda}=\frac{\binom{\nu-t}{k}}{\binom{\nu-t}{k-t}}  \lambda$.
\end{lemma}
\begin{proof}
The desired results follow from Theorem \ref{thm:t to 1}.
\end{proof}

\begin{lemma}\label{lem:wtdTodesign}
Let $\mathcal C$ be a $[\nu, m, d]$ linear code over $\mathrm{GF}(q)$.
 Let $k$ and  $t$ be two positive integers with  $t\le  k \le \nu-t$.
 Suppose that
$ A_k(\mathcal C_T) $  is independent of the specific choice of the elements
in $T$, where $T$ is  any  set of $t$ coordinate positions in $\mathcal  C$. Let
$\overline{\mathcal B}_k(\mathcal C)=\frac{1}{q-1}\{\{\mathcal P(\mathcal C) \setminus \mathrm{Supp}(\mathbf c): \mathbf c \in \mathcal C, \mathrm{wt}(\mathbf c)=k\}\}$.
Then $\left ( \mathcal P(\mathcal C) , \overline {\mathcal B_k}(\mathcal C) \right )$ is a $t$-$(\nu, \nu-k, \overline \lambda)$ design, where $\overline \lambda=A_k(\mathcal C_T)/(q-1)$.
Further, $\left ( \mathcal P(\mathcal C) , \mathcal{B}_k(\mathcal C) \right )$ is a $t$-$(\nu, k,  \lambda)$ design, where
$$
\lambda=\frac{\binom{\nu-t}{\nu-k}A_k(\mathcal{C}_k)}{\binom{\nu-t}{\nu-t-k} (q-1)}.
$$
\end{lemma}

\begin{proof}
Let $T=\{i_1, \ldots, i_t\}$ be a subset of $\mathcal P(\mathcal C)$. Note that
$$
T \subseteq \mathcal P(\mathcal C) \setminus \mathrm{Supp}(\mathbf c) \mbox{ and }
\wt(\mathbf c)=k
$$
if and only if
$$
T \cap \mathrm{Supp}(\mathbf c) = \emptyset \mbox{ and }
\wt(\mathbf c)=k
$$
if and only if
$$
\mathbf{c} \in \mathcal{C}_T \mbox{ and }
\wt(\mathbf c)=k.
$$
By assumption, $T$ is included in $A_k(\mathcal C_T)/(q-1)$ blocks of $\overline {\mathcal B_k}(\mathcal C)$, which is independent of the choices of the elements in $T$. This completes the
proof of the first conclusion. The conclusion of the second part then follows from Lemma
\ref{lem:codeTOdesign}.
\end{proof}

The following theorem gives a characterization of codes supporting $t$-designs via the weight distributions of their shortened and punctured codes.

\begin{theorem}\label{thm:tdesign-wtcode}
Let $\mathcal C$ be a $[\nu, m, d]$ linear code over $\mathrm{GF}(q)$ and  $d^{\perp}$  the minimum distance of $\mathcal  C^{\perp}$.
 Let  $t$ be a positive integer with  $0< t <\min \{d, d^{\perp}\}$. Then the following statements are equivalent.

 (1) $\left ( \mathcal P(\mathcal C) , \mathcal B_k(\mathcal C) \right )$ is a $t$-design for any $0 \le k \le \nu$.

 (2)  $\left ( \mathcal P(\mathcal C^{\perp}) , \mathcal B_k(\mathcal C^{\perp}) \right )$ is a $t$-design for any $0\le k \le \nu$.

 (3) For any $1 \le t' \le t$, the weight distribution $\left ( A_k(\mathcal C_T) \right )_{k=0}^{\nu-t'}$ of the shortened code $\mathcal C_T$
 is independent of the specific choice of the elements in $T$,
 where $T$ is  any  set of $t'$ coordinate positions in $\mathcal  C$.

  (4) For any $1 \le t' \le t$, the weight distribution $\left ( A_k(\mathcal C^T) \right )_{k=0}^{\nu-t'}$ of the punctured code $\mathcal C^T$
  is independent of the specific choice of the elements in $T$,
 where $T$ is  any  set of $t'$ coordinate positions in $\mathcal  C$.
\end{theorem}

\begin{proof}
(3) $\Longrightarrow$ (1): Suppose that the weight distribution $\left ( A_k(\mathcal C_T) \right )_{k=0}^{\nu-t'}$ of the shortened code $\mathcal C_T$
 is independent of the specific choice of the elements in $T$, where $1 \le t'  \le t$.
 By Lemmas \ref{lem:codeTOdesign} and \ref{lem:wtdTodesign}, the pair
 $\left ( \mathcal P(\mathcal C) , \mathcal B_k(\mathcal C) \right )$ is a $t'$-design for any $0 \le k \le \nu-t'$.
 In particular, the pair $\left ( \mathcal P(\mathcal C) , \mathcal B_k(\mathcal C) \right )$ is a $t$-design for any $0 \le k \le \nu-t$
 and $\left ( \mathcal P(\mathcal C) , \mathcal B_k(\mathcal C) \right )$ is a $(\nu -k)$-design for any $\nu-t+1 \le k \le \nu-1$.
 By Lemma \ref{lem:bigblock}, the pair $\left ( \mathcal P(\mathcal C) , \mathcal B_k(\mathcal C) \right )$ is also a $t$-design for any $\nu -t+1 \le k \le \nu-1$.
Since $\left ( \mathcal P(\mathcal C) , \mathcal B_\nu(\mathcal C) \right )$ is always a $t$-design, the pair $\left ( \mathcal P(\mathcal C) , \mathcal B_k(\mathcal C) \right )$
is a $t$-design for any $0 \le k \le \nu$.

(1) $\Longrightarrow$ (4): Recall that if $\left ( \mathcal P(\mathcal C) , \mathcal B_\nu(\mathcal C) \right )$ is  a $t$-design, the pair
$\left ( \mathcal P(\mathcal C) , \mathcal B_\nu(\mathcal C) \right )$ is also  a $t'$-design for $1\le t' \le t$. The desired results follow from
Theorem \ref{thm:pct-code}.

(4) $\Longrightarrow$ (2): By the condition in (4), Lemma \ref{lem:C-S-P} and the Pless power moments in (\ref{eq:PPM}),
the weight distribution $\left ( A_k((\mathcal C^{\perp})_T) \right )_{k=0}^{\nu-t'}$
of the shortened code $(\mathcal C^{\perp})_T$
 is independent of the specific choice of the elements in $T$. Since Statement (3) implies
Statement (1), the desired conclusion then follows.

(2) $\Longrightarrow$ (3): By the condition in Item (2) and Theorem \ref{thm:pct-code}, the weight distribution
$\left ( A_k((\mathcal C^{\perp})^T) \right )_{k=0}^{\nu-t'}$ of the punctured code $(\mathcal C^{\perp})^T$
  is independent of the specific choice of the elements in $T$,
 where $T$ is  any  set of $t'$ coordinate positions in $\mathcal  C^{\perp}$.
 The desired conclusion follows from Lemma \ref{lem:C-S-P} and the Pless power moments in (\ref{eq:PPM}).
\end{proof}

Notice that
some of the $t$-designs $\left ( \mathcal P(\mathcal C) , \mathcal B_k(\mathcal C) \right )$
mentioned in Theorem \ref{thm:tdesign-wtcode} are trivial and some may not be simple.

Theorem \ref{thm:tdesign-wtcode} gives necessary and sufficient conditions for
a code to support $t$-designs with $0< t <\min \{d, d^{\perp}\}$.
It demonstrates the importance of the weight
distribution of linear codes in the theory of $t$-designs, and will be used to develop
a generalisation of the original Assmus-Mattson Theorem in the next section.

The following well-known result is clearly a corollary of Theorem \ref{thm:tdesign-wtcode}.
This demonstrates another usefulness of Theorem \ref{thm:tdesign-wtcode}.

\begin{corollary} \cite[p.165]{MS77}
Let $\mathcal{C}$ be a $[\nu, m, d]$ binary linear code with $m >1$, such that for each $w>0$
the supports of the codewords of weight $w$ form a $t$-design, where $t<d$. Then the supports
of the codewords of each nonzero weight in $\mathcal{C}^\perp$ also form a $t$-design.
\end{corollary}

\section{A generalization of the Assmus-Mattson theorem}\label{sec:geam}

There is a strengthening of the Assmus-Mattson Theorem for special binary codes \cite{Calder91}.
The objective of this section is to present another generalisation of the Assmus-Mattson Theorem
documented in Theorem \ref{thm-AMTheoremExt} and demonstrate its advantages over
the original version.

\subsection{Our generalisation of the Assmus-Mattson theorem}

To develop the generalization of the Assmus-Mattson theorem, we need to prove the following lammas first.

\begin{lemma}\label{lem-C-shortened}
Let $\mathcal C$ be a linear code of length $\nu$ over $\mathrm{GF}(q)$ and  $d^{\perp}$  the minimum distance of $\mathcal  C^{\perp}$.
Let $t$, $k$ be  integers with $0 \le k \le \nu$ and $0< t < \min\{d, d^{\perp}\}$. Let $\left ( \mathcal P(\mathcal C), \mathcal B_{k}(\mathcal C) \right )$
be a $t$-$(\nu, k, \lambda_k)$ design for some integer $\lambda_k$. Let $T$ be  a  set of $t$ coordinate positions in $\mathcal  C$.
Then
 $$A_k(\mathcal C_{T})= \frac{\binom{\nu-t}{k}}{\binom{\nu-t}{k-t}}(q-1) \lambda_k.$$
\end{lemma}

\begin{proof}
Let $\lambda^{T}$ be the number of blocks in $\mathcal B_k(\mathcal C)$ that are disjoint with $T$. Then, $A_k(\mathcal C_{T})=(q-1)\lambda^{T}$.
Using Theorem \ref{thm:t to 1}, one gets
\begin{align*}
A_k(\mathcal C_{T})=(q-1)  \frac{\binom{\nu-t}{k}}{\binom{\nu-t}{k-t}} \lambda_k.
\end{align*}
It completes the proof.
\end{proof}

\begin{lemma}\label{lem:C-DUAL}
Let $\mathcal C$ be a linear code of length $\nu$ over $\mathrm{GF}(q)$ and  $d^{\perp}$  the minimum distance of $\mathcal  C^{\perp}$.
 Let $s$ and $t$ be two positive integers with  $0< t < \min\{d, d^{\perp}\}$.
Let $T$ be  a  set of $t$ coordinate positions in $\mathcal  C$.
Suppose that
 $\left ( \mathcal P(\mathcal C^{\perp}), \mathcal B_{i}(\mathcal C^{\perp}) \right )$
are $t$-$(\nu, i, \lambda^{\perp}_i)$ designs  for all $i$ with  $0\le i \le s+t-1$.
Then
$$A_{k}\left ( (\mathcal C^{\perp})^{T} \right )= (q-1)\sum_{i=0}^t \binom{t}{i} \lambda_{k+i}^{\perp}(t-i,i),$$
where $0\le k \le s-1$ and $\lambda_{k+i}^{\perp}(t-i,i)=\frac{\binom{\nu-t}{k}}{\binom{\nu-t}{k-t+i}} \lambda_{k+i}^\perp$.
\end{lemma}

\begin{proof}
The desired results follow from Lemma \ref{lem:P:k:k+t} and the fact that
$$A_{k+i}(\mathcal C^\perp)=(q-1) \frac{\binom{\nu}{t}}{\binom{k+i}{t}} \lambda_{k+i}^\perp.$$
\end{proof}

\begin{lemma}\label{lem:C-T}
Let $\mathcal C$ be a  $[\nu,m,d]$ code over $\mathrm{GF}(q)$ and  $d^{\perp}$  the minimum distance of $\mathcal  C^{\perp}$.
Let $i_1, \dots, i_s$ be $s$ positive integers and $T$ a set of $t$ coordinate positions of $\mathcal C$,
where $0 \le i_1<\cdots <i_s\le \nu-t$ and  $1\le t<\min \{d, d^{\perp}\}$. Suppose that  $A_i(\mathcal  C_T)~(i\not \in \{ i_1, \dots, i_s \})$ and $A_1(\left ( \mathcal C^{\perp} \right )^{T})$, $\dots$,  $A_{s-1}(\left ( \mathcal C^{\perp} \right )^{T})$
are independent of the elements of $T$.
Then,  the weight distribution  of $\mathcal C_{T}$  is independent of
the elements of $T$ and can be determined from the first $s$ equations in (\ref{eq:PPM}).
\end{lemma}
\begin{proof}
By Lemma \ref{lem:C-S-P}, $\mathcal C_{T}$ has dimension $m-t$, and $\left (\mathcal C_{T} \right )^{\perp} = \left ( \mathcal C^{\perp} \right )^{T}$.
The desired conclusions of this lemma
then follow from Theorem \ref{eq:wt-wtd}.
\end{proof}

One of the main contributions of this paper is the following theorem, which generalizes the Assmus-Mattson theorem.

\begin{theorem}\label{thm-designGAMtheorem}
Let $\mathcal C$ be a linear code over $\mathrm{GF}(q)$ with length $\nu$ and minimum weight $d$.
Let $\mathcal C^{\perp}$ denote the dual code of $\mathcal C$ with minimum weight $d^{\perp}$.
Let $s$ and $t$ be two positive integers with $t< \min \{d, d^{\perp}\}$. Let $S$ be a $s$-subset
of $\{d, d+1, \ldots, \nu-t  \}$.
Suppose that
$\left ( \mathcal P(\mathcal C), \mathcal B_{\ell}(\mathcal C) \right )$ and $\left ( \mathcal P(\mathcal C^{\perp}), \mathcal B_{\ell^{\perp}}(\mathcal C^{\perp}) \right )$
are $t$-designs  for
$\ell    \in \{d, d+1, \ldots, \nu-t  \} \setminus S $ and $0\le \ell^{\perp} \le s+t-1$. Then
 $\left ( \mathcal P(\mathcal C) , \mathcal B_k(\mathcal C) \right )$ and
  $\left ( \mathcal P(\mathcal C^{\perp}), \mathcal B_{k}(\mathcal C^{\perp}) \right )$ are
  $t$-designs for any
$t\le k  \le \nu$, and in particular,
\begin{itemize}
\item $\left ( \mathcal P(\mathcal C) , \mathcal B_k(\mathcal C) \right )$ is a simple $t$-design
      for all $k$ with $d \leq k \leq w$, where $w$ is defined to be the largest integer
      satisfying $w \leq \nu$ and
      $$
      w-\left\lfloor \frac{w+q-2}{q-1} \right\rfloor <d;
      $$
\item  and $\left ( \mathcal P(\mathcal C^{\perp}), \mathcal B_{k}(\mathcal C^{\perp}) \right )$ is
       a simple $t$-design
      for all $k$ with $d \leq k \leq w^\perp$, where $w^\perp$ is defined to be the largest integer
      satisfying $w^\perp \leq \nu$ and
      $$
      w^\perp-\left\lfloor \frac{w^\perp+q-2}{q-1} \right\rfloor <d^\perp.
      $$
\end{itemize}
\end{theorem}

\begin{proof}
For any $1\le t' \le t$, let $S_{t'}=S \cup \{i: \nu -t+1 \le i \le  \nu -t'\}$ and $s'= \#S_{t'}$. Then, $s'=s+t-t'$.
Then, the pair $\left ( \mathcal P(\mathcal C), \mathcal B_{\ell}(\mathcal C) \right )$ is $t'$-design for any $\ell \in \{0, 1, \ldots, \nu-t'\} \setminus S_{t'}$.
By Lemma \ref{lem-C-shortened},  $A_i(\mathcal  C_T)~(i \in \{0, 1, \ldots, \nu-t'\} \setminus S_{t'})$  are independent of the elements of $T$,
where $T$ is any set of $t'$ coordinate positions of $\mathcal C$.

By the assumption of this theorem,  the pair $\left ( \mathcal P(\mathcal C^{\perp}), \mathcal B_{\ell^{\perp}}(\mathcal C^{\perp}) \right )$
is $t'$-design  for $0\le \ell^{\perp}\le (s'+t'-1)=  (s+t-1)$.
By Lemma \ref{lem:C-DUAL},   $A_1\left( ( \mathcal C^{\perp} )^{T}\right )$, $\dots$,  $A_{s'-1}\left(( \mathcal C^{\perp} )^{T}\right)$
are independent of the elements of $T$, where $T$ is any set of $t'$ coordinate positions of $\mathcal C$.

By Lemma \ref{lem:C-T}, the weight distribution  of $\mathcal C_{T}$  is independent of the
choice of the elements of $T$.
It then follows from Theorem \ref{thm:tdesign-wtcode} that
$\left ( \mathcal P(\mathcal C), \mathcal B_k(\mathcal C) \right )$ and
  $\left ( \mathcal P(\mathcal C^{\perp}), \mathcal B_{k}(\mathcal C^{\perp}) \right )$ are
  $t$-designs for any
$t\le k  \le \nu$. The last conclusions on the simplicity of the designs $\left ( \mathcal P(\mathcal C), \mathcal B_k(\mathcal C) \right )$ and
  $\left ( \mathcal P(\mathcal C^{\perp}), \mathcal B_{k}(\mathcal C^{\perp}) \right )$ follow
  from Lemma \ref{lem-simpled}.
\end{proof}

Notice that some of the $t$-designs from Theorem \ref{thm-designGAMtheorem} are trivial,
and some may not be simple. However, many of them are simple and nontrivial, and thus interesting.

We now show that Theorem \ref{thm-AMTheoremExt} (i.e., the Assmus-Mattson Theorem) is a corollary
of Theorem \ref{thm-designGAMtheorem}. To this end, we use Theorem \ref{thm-designGAMtheorem} to
derive Theorem \ref{thm-AMTheoremExt}.

\begin{proof}[Proof of Theorem \ref{thm-AMTheoremExt} using Theorem \ref{thm-designGAMtheorem}]
Let $w_1, w_2, \ldots, w_{s}$ be the nonzero weights of $\mathcal C$ in $\{d, d+1, \ldots, \nu-t\}$,
where $s \leq d^{\perp}-t$. Put $S=\{ w_1, w_2, \ldots, w_{s}\}$. Then $\left ( \mathcal P(\mathcal C), \mathcal B_{\ell}(\mathcal C) \right )$ is
 the trivial $t$-design $\left ( \mathcal P(\mathcal C), \emptyset \right )$
 for all $\ell \in \{d,d+1, \ldots, \nu-t\} \setminus S$. Note
that $s+t-1 \leq d^{\perp}-1$. Clearly, $\left ( \mathcal P(\mathcal C^{\perp}), \mathcal B_{\ell^{\perp}}(\mathcal C^{\perp}) \right )$
are the trivial $t$-design $\left ( \mathcal P(\mathcal C^{\perp}), \emptyset \right )$ for all $0 \le \ell^{\perp} \le s+t-1$.
 It then follows from Theorem \ref{thm-designGAMtheorem} that
 $\left ( \mathcal P(\mathcal C) , \mathcal B_k(\mathcal C) \right )$ and
  $\left ( \mathcal P(\mathcal C^{\perp}), \mathcal B_{k}(\mathcal C^{\perp}) \right )$ are
  $t$-designs for any
$t\le k  \le \nu$. Both $\left ( \mathcal P(\mathcal C) , \mathcal B_k(\mathcal C) \right )$ and
  $\left ( \mathcal P(\mathcal C^{\perp}), \mathcal B_{k}(\mathcal C^{\perp}) \right )$  are clearly the trivial design
$\left ( \mathcal P(\mathcal C) , \emptyset \right )$ for $0 \leq k \leq t-1$, as we assumed that $t < \min\{d, d^\perp\}$.
The desired conclusions of Theorem \ref{thm-AMTheoremExt} then follow.
\end{proof}

One would naturally ask if Theorem \ref{thm-designGAMtheorem} is more powerful than
Theorems \ref{thm-AMTheoremExt} and \ref{thm-AMTheorem}. The answer is yes, and
this will be justified in the next subsection.

\subsection{The extended Assmus-Mattson theorem can outperform the origianl one}\label{sec-examplespec}

The objective of this section is to show that Theorem \ref{thm-designGAMtheorem} is
more powerful than Theorems \ref{thm-AMTheoremExt} and \ref{thm-AMTheorem}, and
is indeed useful. To this end, we consider the linear codes investigated in
\cite{DMT18} and \cite{TDX19}.

In order for Theorem \ref{thm-designGAMtheorem} to outperform the original Assmus-Mattson
Theorem,
one has to choose two positive integers $s$ and $t$
with $t < \min \{d, d^{\perp}\}$ and an $s$-subset $S$ of $\{d, d+1,   \ldots, \nu-t  \}$, and
then prove that $\left ( \mathcal P(\mathcal C), \mathcal B_{\ell}(\mathcal C) \right )$ and
$\left ( \mathcal P(\mathcal C^{\perp}), \mathcal B_{\ell^{\perp}}(\mathcal C^{\perp}) \right )$
are $t$-designs for $\ell  \in \{d, d+1,   \ldots, \nu-t  \} \setminus S$ and $0 \le \ell^{\perp} \le s+t-1$ with some other approach. Hence, extra work is needed when applying Theorem \ref{thm-designGAMtheorem}. This
intuitively explains why Theorem \ref{thm-designGAMtheorem} can outperform the original Assmus-Mattson
The following two examples will clarify this statement.

\begin{example}\label{exam-spec1}
Let $F$ be a bent vectorial  function from $\mathrm{GF}(2^{2m})$ to $\mathrm{GF}(2^\ell)$, where $m\ge 3$.
Let $\mathcal C(F)$ be the code given in (\ref{eq:d-c(vb)}). By the weight distribution of $\mathcal C(F)$ in Table \ref{table:supp-bentcode},
for $k\not \in \{2^{2m-1}, 2^{2m-1}\pm 2^{m-1}\}$, the pair $(\mathcal P(\mathcal C(F)), \mathcal B_k(\mathcal C(F)))$
is a  trivial $2$-design. By the definition of $\mathcal C(F)$, one has
 $ \mathcal B_{2^{2m-1}}(\mathcal C(F))=  \mathcal B_{2^{m-1}}(\mathrm{RM}_2(1,2m))$, where
 $\mathrm{RM}_2(1,2m)$ is the first order  Reed-Muller code given by
 \begin{align*}
 \mathrm{RM}_2(1,2m)=\left \{(\mathrm{Tr}(bx)+c)_{x\in \mathrm{GF}(2^{2m})}: b\in \mathrm{GF}(2^{2m}), c\in \mathrm{GF}(2)\right \}.
 \end{align*}
It is well known that $ \mathcal B_{2^{m-1}}(\mathrm{RM}_2(1,2m))$ holds $2$-design.
Let $S=\{2^{2m-1}+ 2^{m-1}, 2^{2m-1}- 2^{m-1}\}$. Then, the pair $(\mathcal P(\mathcal C(F)), \mathcal B_k(\mathcal C(F)))$
is a   $2$-design for any $k\in \{0,1,\ldots, 2^{2m}-2\} \setminus S$. Since $d((\mathcal C(F))^{\perp})=4$, the pair
$(\mathcal P(\mathcal C(F)^{\perp}), \mathcal B_k(\mathcal C(F)^{\perp}))$
is a trivial $2$-design for $0\le k \le 3= \#S +2 -1$. Hence, by Theorem \ref{thm-designGAMtheorem}, the codes $\mathcal C(F)$ and $\mathcal C(F)^{\perp}$
support $2$-designs \cite[Theorem 11]{DMT18}. The weight distribution of
the code $\mathcal C(F)$ and Lemma \ref{lem-simpled} tell us that the $2$-designs
supported by $\mathcal C(F)$ are simple.
\end{example}

\begin{example}\label{exam-spec2}
Let $m$ be an odd positive integer. Let $\mathcal C$  be the  linear code defined by
\begin{align*}
\mathcal{C}=\left \{ \left (\mathrm{Tr}_{3^m/3}\left (a \alpha^{4i}+b \alpha^{2i} \right ) \right )_{i=0}^{\frac{3^m-1}{2}-1} :a,b\in \mathrm{GF}(3^m)\right  \},
\end{align*}
where $\mathrm{Tr}_{3^m/3}(\cdot)$ is the trace function from $\mathrm{GF}(3^m)$ to $\mathrm{GF}(3)$ and $\alpha$ is a generator of $\mathrm{GF}(3^m)^*$.
Then the code $\mathcal C$  have parameters $[\frac{3^m-1}{2}, 2m, 3^{m-1}-3^{\frac{m-1}{2}}]$.
Let $S=\left \{3^{m-1},  3^{m-1}\pm 3^{\frac{m-1}{2}} \right \}$. Then,
$A_k(\mathcal C)=0$ if $k \not\in S \cup \{0\}$. Thus, the pair $(\mathcal P(\mathcal C), \mathcal B_k(\mathcal C))$
is a trivial    $2$-design for any $k\in \{0,1,\ldots, \frac{3^m-1}{2} -2\} \setminus S$.
According to \cite[Corollary 1]{TDX19}, $(\mathcal P(\mathcal C^{\perp}), \mathcal B_4(\mathcal C^{\perp}))$ is a Steiner system $S(2,4,\frac{3^m-1}{2})$ and is simple.
It was known that $d(\mathcal C^{\perp})=4$ \cite{TDX19}.
Thus the pair  $(\mathcal P(\mathcal C^{\perp}), \mathcal B_4(\mathcal C^{\perp}))$ is a $2$-design for $0 \le k \le 4=\#S+2-1$.
Hence, by Theorem \ref{thm-designGAMtheorem}, the codes $\mathcal C$ and $\mathcal C^{\perp}$
support $2$-designs  \cite[Theorems 11 and 12]{TDX19}.  The weight distribution of
the code $\mathcal C$ and Lemma \ref{lem-simpled} tell us that the $2$-designs
supported by $\mathcal C(F)$ are simple.
\end{example}

The weight distributions of the codes in Examples \ref{exam-spec1} and \ref{exam-spec2}
and the minimum distances of their duals are known.
They tell us that the original Assumus-Mattson Theorems (i.e, Theorems \ref{thm-AMTheoremExt} and
\ref{thm-AMTheorem}) cannot be applied to prove
that the codes in Examples \ref{exam-spec1} and \ref{exam-spec2} support
$2$-designs. It is also known that the automorphism groups of these codes
are not 2-transitive in general \cite{DMT18,TDX19}. However, Theorem \ref{thm-designGAMtheorem}
can do it. Therefore, Theorem \ref{thm-designGAMtheorem} is more powerful than
Theorems \ref{thm-AMTheoremExt} and \ref{thm-AMTheorem}.
Another application of Theorem \ref{thm-designGAMtheorem} will be given
in the next section.

\section{$2$-designs and differentially $\delta$-uniform functions}\label{sec:dede}

Recall the definition of differentially $\delta$-uniform functions over $\mathrm{GF}(2^n)$ and the notation introduced in Section \ref{sec-duf}.
In this section, we shall give a connection between differentially $\delta$-uniform functions and $2$-designs, and present some new $2$-designs from some special differentially two-valued functions.

Let $F$ be  a differentially $\delta$-uniform function over $\mathrm{GF}(2^n)$.
Define the following linear code
\begin{align*}
\mathcal C(F)=\left \{ \left (\mathrm{Tr}(aF(x)+bx)+c \right )_{x\in \mathrm{GF}(2^n)}:a,b \in \mathrm{GF}(2^n), c\in \mathrm{GF}(2) \right \}.
\end{align*}
It follows from Delsarte's theorem \cite{MS77} that the dual code $\mathcal C(F)^{\perp}$ of $\mathcal C(F)$ can be given by
\begin{align*}
\mathcal C(F)^{\perp}=\left \{ (c_x)_{x\in \mathrm{GF}(2^n)}   \in \mathrm{GF}(2)^{n}: \sum_{x\in \mathrm{GF}(2^n)} c_x \mathbf u_x=\mathrm{0}\right \},
\end{align*}
where $\mathbf u_x=(F(x), x, 1)$.
For any $x_1, x_2 \in \mathrm{GF}(2^n)$ with $x_1\neq x_2$, denote by  $\lambda_{\{x_1, x_2\}}$ the cardinality  of the set
 $$W_{\{x_1, x_2\}}=\left \{ \mathbf c= (c_x)_{x\in \mathrm{GF}(2^n)} \in \mathcal C(F)^{\perp}: \mathrm{wt}(\mathbf c)=4, c_{x_1}=c_{x_2}=1 \right  \}.$$
Let $a=x_1+x_2$ and $b=F(x_1)+F(x_2)$. Denote $E_{\{x_1, x_2\}}=\{x\in \mathrm{GF}(2^n): F(x+a)+F(x)=b\}$.
Then, $\delta(a,b)=\#(E_{\{x_1, x_2\}})$ and
\begin{align*}
E_{\{x_1, x_2\}}=   \{x_1, x_2\} \cup \left (\cup_{i=1}^{\delta(a,b)/2-1} \{x_i', x_i'+a\} \right ),
\end{align*}
where $x_i'\in \mathrm{GF}(2^n)$.
Moreover,  it is easily observed that
\begin{align*}
W_{\{x_1, x_2\}}=\{\mathbf c_i: 1 \le i \le \delta(a, b)/2-1 \},
\end{align*}
where $\mathbf c_i =(c_x)_{x\in \mathrm{GF}(2^n)}$ with
\begin{align*}
c_x=\left\{
  \begin{array}{ll}
    1, & x \in \{x_i', x_i'+a, x_1, x_2\}; \\
    0, & \hbox{otherwise.}
  \end{array}
\right.
\end{align*}
Consequently, one has
\begin{align*}
\lambda_{\{x_1, x_2\}}= \frac{\delta\left (x_1+x_2, F(x_1)+F(x_2) \right )-2}{2}.
\end{align*}

So, we have proved the following theorem, which establishes a link between some $2$-designs and  differentially two-valued functions.

\begin{theorem}\label{thm:design-diff}
Let $F(x)$ be  a  function over $\mathrm{GF}(2^n)$. Then $\left ( \mathcal P(\mathcal C(F)^{\perp}),  \mathcal B_4(\mathcal C(F)^{\perp})
\right )$ is a $2$-design if and only if $F$  is differentially two-valued.
Furthermore, if $F$  is differentially two-valued with $\{0, 2^s\}$, then $\left ( \mathcal P(\mathcal C(F)^{\perp}),  \mathcal B_4(\mathcal C(F)^{\perp})
\right )$ is a $2$-$(2^n, 4, 2^{s-1}-1)$ design.
\end{theorem}

\begin{corollary}
Let $F(x)$ be  a  function over $\mathrm{GF}(2^n)$. Then
$\left ( \mathcal P(\mathcal C(F)^{\perp}),  \mathcal B_4(\mathcal C(F)^{\perp})
\right )$ is  a Steiner
system $S(2, 4, 2^n)$ if and only if $F$  is differentially two-valued with $\{0, 4\}$.
\end{corollary}

Magma program shows that the Steiner system $S(2,4,2^n)$ from
the differentially two-valued $\{0, 4\}$ function  $F(x)=x^{2^{2i}-2^i+1}$ \cite{BCC10,HP08} or $F(x)=\alpha x^{2^i+1}+ \alpha^{2^m} x^{2^{2m}+2^{m+i}}$ \cite{BTT12}
 is equivalent to the incidence structure from points and lines of the affine geometry $\mathrm{AG}(2^{\frac{n}{2}},\mathrm{GF}(4))$.
It is still open whether  there is a differentially two-valued $\{0, 4\}$ function $F(x)$ such that
$(\mathcal P(\mathcal C(F)^{\perp}), \mathcal B_4(\mathcal C(F)^{\perp}))$ is not equivalent to the Steiner system from affine geometry.

With Theorem \ref{thm:design-diff}, we can directly use results of the differentially two-valued functions to study  the incidence  structure
$\left ( \mathcal P(\mathcal C(F)^{\perp}),  \mathcal B_4(\mathcal C(F)^{\perp})
\right )$. By Lemma 1 in \cite{CP19} and Theorem \ref{thm:design-diff}, one has the following.

\begin{corollary}
Let $F(x)$ be  a differentially $\delta$-uniform function over $\mathrm{GF}(2^n)$.
Then
$\left ( \mathcal P(\mathcal C(F)^{\perp}),  \mathcal B_4(\mathcal C(F)^{\perp})
\right )$ forms a $2$-design if and only if
\begin{align*}
\sum_{(a , b) \in \mathrm{GF}(2^n)^* \times \mathrm{GF}(2^n)} \mathcal W_{F}(a, b)^4= 2^{2n}(2^n-1) \delta.
\end{align*}
\end{corollary}

\begin{theorem}\label{thm:wtd-ds-dg}
Let $F(x)$ over $\mathrm{GF}(2^n)$ be differentially two-valued with $\{0, 2^s\}$. Suppose that
$\{ \mathcal W_{F}(\lambda, \mu): \lambda \in \mathrm{GF}(2^n)^*, \mu \in \mathrm{GF}(2^n)  \}
=  \{0,  2^{\frac{n+s}{2}}, -2^{\frac{n+s}{2}}\}$.
Then, the code $\mathcal C(F)$ and its dual $\mathcal C(F)^{\perp}$ support $2$-designs.
\end{theorem}
\begin{proof}
Let $S=\left \{2^{n-1}, 2^{n-1} \pm 2^{\frac{n+s-2}{2}} \right \}$.
Since $\mathcal W_{F}(\lambda, \mu)\in \{0,  2^{\frac{n+s}{2}}, -2^{\frac{n+s}{2}}\}$, the  incidence structure
$\left ( \mathcal P(\mathcal C(F)),  \mathcal B_k(\mathcal C(F))
\right )$ forms a trivial $2$-design for any $k \not \in S$.
It follows from Theorem \ref{thm:design-diff} and $d(\mathcal C(E)^{\perp})\ge 4$ \cite[Theorem 9]{CCZ98} that the incidence structure
$\left ( \mathcal P(\mathcal C(F)^{\perp}),  \mathcal B_k(\mathcal C(F)^{\perp})
\right )$ forms a $2$-design for $0 \le k \le 4= \# S+2-1$.
The desired conclusions then follow from Theorem \ref{thm-designGAMtheorem}.
\end{proof}

\begin{corollary}
Let $q$ be a power of $2$ and $m$ be a positive integer. Let $F(x)$ be a quadratic permutation over $\mathrm{GF}(q^m)$ of the form
\begin{eqnarray} \nonumber
F(x)=\sum_{0\leq i\leq j\leq m-1}c_{ij}x^{q^i+q^j}, \quad  \forall c_{ij} \in \mathrm{GF}({q^m}).
\end{eqnarray}
Suppose that $F(x)$ is  differentially $q$-uniform. Then, the code $\mathcal C(F)$ and its dual $\mathcal C(F)^{\perp}$ support $2$-designs.
\end{corollary}
\begin{proof}
By Theorems 5 and 6  in \cite{MTX19}, the function $F(x)$ is differentially two-valued with
$\{0, q\}$ and has Walsh coefficients in $\{0, \pm q^{\frac{m+1}{2}}\}$.
The desired conclusion then follows from Theorem \ref{thm:wtd-ds-dg}.
\end{proof}

To determine the parameters of the $2$-designs from the
code $\mathcal C(F)$ and its dual $\mathcal C(F)^{\perp}$, we need the following lemma.
\begin{lemma}\label{lem:wtd-dgpra}
Let $F(x)$  be a  function over $\mathrm{GF}(2^n)$ with Walsh coefficients in $\{0, \pm 2^{\frac{n+s}{2}}\}$, where $0\le s \le n-1$.
Then the code $\mathcal C(F)$ has parameters $[2^n, 2n+1, 2^{n-1}-2^\frac{n+s-2}{2}]$ and its dual code $\mathcal C(F)^{\perp}$ has minimum distance
\begin{align*}
d^{\perp}=\begin{cases}
4, & s\ge 2, \\
6, & s=1.
\end{cases}
\end{align*}
Furthermore, the weight distribution of $\mathcal C(F)$ is given by
\begin{align*}
&A_{2^{n-1}- 2^{\frac{n+s-2}{2}}}=2^{n-s}(2^n-1),\\
&A_{2^{n-1}}= (2^n-1)(2^{n+1}-2^{n-s+1}+2),\\
&A_{2^{n-1}+ 2^{\frac{n+s-2}{2}}}=2^{n-s}(2^n-1),\\
&A_{2^n}=1,
\end{align*}
and $A_i=0$ for all other $i$. The number $A_{4}^{\perp}$ of the codewords of weight $4$  in $\mathcal C(F)^{\perp}$
equals to  $\frac{2^{n-2}(2^n-1)(2^{s-1}-1)}{3}$.
\end{lemma}
\begin{proof}
Let $\mathbf{c}(a,b,c)=\left (\mathrm{Tr}\left ( aF(x)+bx \right )+c \right )_{x\in \mathrm{GF}(2^n)}$,
where $a,b\in \mathrm{GF}(2^n)$ and $c\in \mathrm{GF}(2)$.
Then
\begin{align*}
\mathrm{wt}(\mathbf{c}(a,b,c))=&\frac{1}{2}\sum_{x\in \mathrm{GF}(2^n)} \left(1 - (-1)^{\mathrm{Tr}\left ( aF(x)+bx \right )+c } \right )\\
=& 2^{n-1}-\frac{(-1)^c}{2} \sum_{x\in \mathrm{GF}(2^n)} (-1)^{\mathrm{Tr}\left ( aF(x)+bx \right ) }\\
=&\begin{cases}
2^{n-1}-\frac{1}{2} \mathcal W_{F}(a, b), & a\neq 0,
\cr 2^{n-1}, & a=0, b\neq 0,
\cr 2^{n}, & a=b=0, c=1,
\cr 0, & a=b=0.
\end{cases}
\end{align*}
Then, $\mathrm{wt}(\mathbf{c}(a,b,c))\in \{0, 2^n, 2^{n-1}, 2^{n-1}\pm 2^{\frac{n+s-2}{2}}\}$,
and $\mathrm{wt}(\mathbf{c}(a,b,c))=0$ if and only if $a=b=c=0$.
Thus, the dimension of $\mathcal C(F)$ is equal to $2n+1$.
By Theorem 9 in \cite{CCZ98}, the minimal distance  $d^{\perp} =4 \text{ or } 6$.
Let $i_1=2^{n-1}- 2^{\frac{n+s-2}{2}}$, $i_2=2^{n-1}$, and $i_3=2^{n-1}+ 2^{\frac{n+s-2}{2}}$.
Note that $A_{2^{n}}=1$. The first three Pless power moments in (\ref{eq:PPM}) give
\begin{align*}
\left\{
  \begin{array}{l}
    A_{i_1} + A_{i_2} +A_{i_3} = 2^{2n+1}-2,   \\
    i_1 A_{i_1} + i_2  A_{i_2} + i_3 A_{i_3}  =   2^{2n+1-1} \cdot 2^n -2^n,  \\
   i_1^2 A_{i_1} + i_2^2  A_{i_2} + i_3^2 A_{i_3}  =2^{2n+1-2} \cdot 2^n(2^n+1)-2^{2n}.
  \end{array}
\right.
\end{align*}
Solving this system of equations, one gets
\begin{align*}
&A_{2^{n-1}- 2^{\frac{n+s-2}{2}}}=2^{n-s}(2^n-1),\\
&A_{2^{n-1}}= (2^n-1)(2^{n+1}-2^{n-s+1}+2),\\
&A_{2^{n-1}+ 2^{\frac{n+s-2}{2}}}=2^{n-s}(2^n-1).
\end{align*}
Using the fourth Pless power moment in (\ref{eq:PPM}), one has
\begin{align*}
A_4^{\perp}=\frac{2^{n-2}(2^n-1)(2^{s-1}-1)}{3}.
\end{align*}
Since $d^{\perp}=4$ or $6$, one obtains
\begin{align*}
d^{\perp}=\begin{cases}
4, & s\ge 2,
\cr 6, & s=1.
\end{cases}
\end{align*}
It completes the proof.
\end{proof}

Combining Theorem \ref{thm:wtd-ds-dg} and Lemma \ref{lem:wtd-dgpra}, we deduce the following.
\begin{theorem}
Let $F(x)$ over $\mathrm{GF}(2^n)$ be differentially two-valued with $\{0, 2^s\}$
and have
Walsh coefficients in  $ \{0, 2^{\frac{n+s}{2}}, -2^{\frac{n+s}{2}}\}$.
Then, $\mathcal C(F)$ holds a  $2$-$(2^n,k,\lambda)$ design for the
following pair $(k, \lambda)$:
\begin{itemize}
\item $(k, \lambda)=\left (2^{n-1}\pm  2^{\frac{n+s-2}{2}},\left ( 2^{n-s-1}\pm 2^{\frac{n-s-2}{2}} \right ) \left ( 2^{n-1}\pm 2^{\frac{n+s-2}{2}} -1\right ) \right )$, and
\item $(k, \lambda)=\left (2^{n-1}, (2^{n-1}-1)(2^{n}-2^{n-s}+1) \right )$.
\end{itemize}
\end{theorem}

To show the existence of the $2$-designs in Theorem \ref{thm:wtd-ds-dg}, we describe some functions over $\mathrm{GF}(2^n)$ which are differentially two-valued with $\{0, 2^s\}$
and have Walsh coefficients in $\{0, \pm 2^{\frac{n+s}{2}}\}$.

\begin{enumerate}
\item The first family of differentially two-valued   monomials with Kasami exponents: $F(x)=x^{2^{2i}-2^i+1}$, where  $n$ and $i$ are positive integers, $n\neq 3i$,  $s=\gcd(n,i)$, and  $\frac{n}{s}$ is odd.
Then  $F(x)$ is over $\mathrm{GF}(2^n)$ and differentially two-valued with $\{0, 2^s\}$,
and has  Walsh coefficients in $\{0, \pm 2^{\frac{n+s}{2}}\}$ \cite{BCC10,HP08}.
\item The second family of differentially two-valued functions  discovered by  Bracken, Tan, and Tan \cite{BTT12}:
$F(x)=\alpha x^{2^i+1}+ \alpha^{2^m} x^{2^{2m}+2^{m+i}}$, where $n=3m$,  $m$ and $i$ are two positive integers,
$3\nmid m$, $3|(m+i)$,  $s= \gcd(m, i)$,  $2 \nmid \frac{m}{s}$, and  $\alpha$ is a primitive element of $\mathrm{GF}(2^{n})$.
Then  $F(x)$ is over $\mathrm{GF}(2^{n})$ and differentially two-valued with $\{0, 2^s\}$,
and has Walsh coefficients in $\{0, \pm 2^{\frac{3m+s}{2}}\}$ \cite{BTT12}.
\end{enumerate}

When $s\ge 2$, the original Assmus-Mattson Theorem says that the codes $\mathcal C(F)$
and $\mathcal C(F)^{\perp}$ for $F(x)=x^{2^{2i}-2^i+1}$
and $F(x)=\alpha x^{2^i+1}+ \alpha^{2^m} x^{2^{2m}+2^{m+i}}$ support only $1$-designs.
Magma program shows that, in general, the codes $\mathcal C(F)$ and $\mathcal C(F)^{\perp}$
are not $2$-transitive or $2$-homogeneous.
However, by our generalization of the Assmus-Mattson theorem, these codes support $2$-designs.
This is the third example showing that Theorem \ref{thm-designGAMtheorem} is more powerful
than the original Assmus-Mattson Theorems (i.e.,
Theorems \ref{thm-AMTheoremExt} and \ref{thm-AMTheorem}).

\section{Summary and concluding remarks}\label{sec:conc}

The main contributions of this paper are the following:
\begin{itemize}
\item The first one is the general theory for punctured and shorted codes of linear codes
      supporting $t$-design documented in Section \ref{sec-pscode1}. The general theory
      led to several classes of binary codes with interesting parameters and known weight
      distributions, which were presented in Sections \ref{sec-pscode2} and \ref{sec-pscode3}.
      Some of the codes are distance-optimal and some have the best known parameters. These
      codes can be used for secret sharing \cite{ADHK98,YD06}.
      The general theory also played an important role in later sections.
\item The second is the characterization of $t$-designs supported by a linear code via the
      weight distributions of punctured and shortened codes of the code, which was documented
      in Theorem \ref{thm:tdesign-wtcode}. This characterization shows the importance of the
      weight distribution of linear codes in constructing $t$-designs from linear codes.
\item The third is the generalized Assmus-Mattson theorem described in Theorem
      \ref{thm-designGAMtheorem}, which outperformed the original Assmus-Mattson Theorem
      in the three cases treated in this paper.
\item The fourth is the link between some $2$-designs and differentially $\delta$-uniform
      functions, which was presented in Section \ref{sec:dede}. With is link, some $2$-designs
      and some Steiner systems $S(2, 4, 2^n)$ were constructed.
\end{itemize}

It would be interesting to use the generalized Assmus-Mattson theorem (i.e., Theorem
\ref{thm-designGAMtheorem}) to obtain more $t$-designs that cannot be produced with the original
Assmus-Mattson theorem. The three cases dealt with in this paper are the only known ones to the
best knowledge of the authors.


\begin{thebibliography}{99}

\bibitem{AK92}
{\sc E. F. Assmus Jr. and J. D. Key}, \textit{Designs and Their Codes,} Cambridge University Press, Cambridge, 1992.

\bibitem{AK98}
{\sc E. F. Assmus Jr. and J. D. Key},   \textit{Polynomial codes and finite geometries,}
in  Handbook of Coding Theory, vol. II,   V. S. Pless and W. C. Huffman, eds.,  Elsevier, Amsterdam, 1998, pp. 1268--1343.

\bibitem{AM69}
{\sc E. F. Assmus  Jr. and H. F. Mattson  Jr.,} \textit{New 5-designs,} J. Combin. Theory  6 (1969), pp. 122--151.

\bibitem{AM74}
{\sc E. F. Assmus Jr. and H.F. Mattson Jr.,}  \textit{Coding and combinatorics,}
SIAM Rev. 16 (1974), pp. 349--388.

\bibitem{ADHK98}
{\sc R. Anderson, C. Ding, T. Helleseth and T. Kl{\o}ve,}   \textit{How to build robust shared control systems,}
 Des. Codes Cryptogr. 15 (1998), pp.  111--124.

\bibitem{BCC10}
{\sc C. Blondeau, A. Canteaut and P. Charpin,}  \textit{Differential properties of power functions,}
In Proceedings of the 2010 IEEE International Symposium on Information Theory, ISIT 10, Austin, USA, June 2010, pp. 2478--2482.

\bibitem{BTT12}
{\sc C. Bracken, C. H. Tan and Y. Tan,} \textit{Binomial differentially $4$-uniform permutations with high nonlinearity,}
Finite Fields Appl. 18 (2012), pp. 537--546.

\bibitem{Calder91}
{\sc A. R. Calderbank, P. Delsarte and N. J. A. Sloane,} \textit{A strengthening of the Assmus-Mattson Theorem,}
IEEE Trans. Inf. Theory 37 (1991), pp. 1261--1268.

\bibitem{CCZ98}
{\sc C. Carlet, P. Charpin and V. Zinoviev,}
\textit{Codes, bent functions and permutations suitable for DES-like cryptosystems,}
Des. Codes Cryptogr. 15 (1998),  pp. 125--156.

\bibitem{CP}
{\sc P. Charpin and J. Peng,} \textit{Differential uniformity and the associated codes of cryptographic functions,}
Advances in Mathematics of Communications, AIMS, in press.

\bibitem{CP19}
{\sc P. Charpin and J. Peng,}  \textit{New links between nonlinearity and differential uniformity,}
Finite Fields Appl. 56 (2019), pp. 188--208.

\bibitem{Ding15}
{\sc C. Ding,}  \textit{Codes from Difference Sets,} World Scientific, Singapore, 2015.

\bibitem{Ding18dcc}
{\sc C. Ding,} \textit{Infinite families of 3-designs from a type of five-weight code,}
Des. Codes Cryptogr. 86 (2018),  pp. 703--719.

\bibitem{Ding18jcd}
{\sc C. Ding,}  \textit{An infinite family of Steiner systems from cyclic codes,}
Journal of Combinatorial Designs 26 (2018), pp. 127--144.

\bibitem{Dingbk18}
{\sc C. Ding,} \textit{Designs from Linear Codes,} World Scientific, Singapore, 2018.

\bibitem{DL17}
{\sc C. Ding and C. Li,} \textit{Infinite families of 2-designs and 3-designs from linear codes,}
Discrete Math. 340 (2017), pp. 2415--2431.

\bibitem{DLX17}
 {\sc C. Ding, C. Li and Y. Xia},
 \textit{Another generalization of the binary Reed-Muller codes and its applications,} Finite Fields Appl. 53 (2018), pp. 144--174.

\bibitem{DMT18}
{\sc C. Ding, A. Munemasa and V. D. Tonchev,} \textit{Bent vectorial functions, codes and designs,}
 IEEE Trans. Inf. Theory, doi: 10.1109/TIT.2019.2922401.

\bibitem{HKM04}
{\sc M. Harada, M. Kitazume and A. Munemasa,} \textit{On a $5$-design related to an extremal doubly-even self-dual code of length $72$,}  J. Combin. Theory, Ser. A 107 (2004), pp. 143--146.

\bibitem{HMT05}
{\sc M. Harada, A. Munemasa and V. D.  Tonchev,}  \textit{A characterization of designs related to an extremal doubly-even self-dual code of length $48$,} Annals of Combinatorics 9 (2005), pp. 189--198.

\bibitem{HP08}
{\sc D. Hertel and A. Pott,} \textit{Two results on maximum nonlinear functions,}
 Des. Codes Cryptogr. 47 (2008), pp. 225--235.

\bibitem{HP10}
{\sc W. C. Huffman and V. Pless,} \textit{ Fundamentals of Error-Correcting Codes,} Cambridge University Press, Cambridge, 2003.

\bibitem{KM00}
{\sc J. H. Koolen and A. Munemasa,}  \textit{Tight 2-designs and perfect $1$-codes in Doob graphs,}
J. Stat. Planning and Inference 86 (2000),  pp. 505--513.

\bibitem{MS77}
{\sc F. J. MacWilliams and N. J. A. Sloane,}  \textit{The Theory of Error-Correcting Codes,}
North-Holland, Amsterdam, 1977.

\bibitem{MTX19}
{\sc S. Mesnager, C. Tang and M. Xiong,} On the boomerang uniformity of quadratic permutations over $ \mathbb F_{2^n} $, preprint available at
https://eprint.iacr.org/2019/277.pdf.


\bibitem{MT04}
{\sc A. Munemasa and V. D. Tonchev,} \textit{A new quasi-symmetric $2$-$(56,16,6)$ design obtained from codes,}
 Discrete Math. 284 (2004), pp.  231--234.

\bibitem{Sti08}
{\sc D. R. Stinson,}  \textit{Combinatorial designs: constructions and analysis,} Sigact News 39 (2008), pp. 17--21.

\bibitem{TDX19}
{\sc C. Tang, C. Ding and M.  Xiong,}  \textit{Steiner systems $ S (2, 4,\frac {3^ m-1}{2}) $ and $2 $-designs from ternary linear codes of length $\frac {3^ m-1}{2}$,}  Des. Codes Cryptogr.,  doi: 10.1007/s10623-019-00651-8.

\bibitem{Ton98}
{\sc V. D. Tonchev,} \textit{Codes and designs,} In Handbook of Coding Theory, vol. II, V. S. Pless and W. C. Huffman, eds., Elsevier, Amsterdam, 1998, pp. 1229--1268. .

\bibitem{Ton07}
{\sc V. D. Tonchev,} \textit{Codes,} In Handbook of Combinatorial Designs, 2nd edition, C. J. Colbourn and J. H. Dinitz,  eds., CRC Press, New York, 2007, pp. 677--701.

\bibitem{YD06}
{\sc J. Yuan and C. Ding,} \textit{Secret sharing schemes from three classes of linear codes,} IEEE Trans. Inf. Theory 52 (2006), pp. 206--212.

\end{thebibliography}
\end{document}